\documentclass{article}

\usepackage{arxiv}

\usepackage[utf8]{inputenc}
\usepackage[T1]{fontenc}
\usepackage{hyperref}
\hypersetup{colorlinks=true, linkcolor=blue!45!black,
  citecolor=blue!45!black, urlcolor=blue!55!black}
\usepackage{url}
\usepackage{amsmath,amssymb,amsthm}
\usepackage{nicefrac}
\usepackage{microtype}
\usepackage{graphicx}
\usepackage{xcolor}
\usepackage{booktabs}
\usepackage{float}
\usepackage{caption}
\usepackage[numbers,sort&compress]{natbib}

\newtheorem{theorem}{Theorem}
\newtheorem{lemma}{Lemma}
\newtheorem{corollary}{Corollary}
\newtheorem{proposition}{Proposition}
\theoremstyle{definition}
\newtheorem{definition}{Definition}

\newcommand{\R}{\mathbb{R}}
\newcommand{\Z}{\mathbb{Z}}
\newcommand{\Stwo}{\mathbb{S}^2}
\newcommand{\kappag}{\kappa_g}
\newcommand{\norm}[1]{\left\|#1\right\|}
\newcommand{\abs}[1]{\left|#1\right|}
\newcommand{\Traw}{T_{\mathrm{raw}}}
\newcommand{\Tf}{T_{\mathrm{CV}}}
\newcommand{\Sdd}{S_{\mathrm{DD6}}}
\newcommand{\Sdgl}{S_{\mathrm{DGL}}}
\newcommand{\Smid}{S_{\mathrm{mid}}}
\newcommand{\kERS}{\kappa}
\newcommand{\kN}{\kappa_{N}}
\newcommand{\kT}{\kappa_{T}}
\newcommand{\doi}[1]{\href{https://doi.org/#1}{doi.org/#1}}

\title{Smooth Curves from Curvature-Driven Subdivision}

\author{
  Hassan Ugail \\
  Centre for Visual Computing and Intelligent Systems \\
  University of Bradford, United Kingdom \\
}

\begin{document}
\maketitle

\begin{abstract}
We introduce and analyse an interpolatory subdivision scheme in which
each inserted point is defined by a curvature-profile model rather than
by an affine combination of its neighbours, and we establish its
convergence and geometric regularity. On each edge, curvatures
estimated from circumscribed circles prescribe a curvature profile, the
segment is reconstructed by Frenet integration with geodesic shooting,
and the new point is taken at half arc length. The construction works
identically in the plane and on the sphere, reproduces geodesics and
circles exactly, and a unique curvature prefilter turns its planar
linearisation into the six-point Deslauriers--Dubuc scheme. The
straightening condition of Ewald, Reif, and Sabin fails for this
scheme. Measuring only the normal part of the relative distortion
repairs their theory, and we prove that the refined polygons converge
to a regular interpolatory limit curve with continuous tangent
direction and curvature, twice continuously differentiable in arc
length. Regularity is thus governed by the normal companion scheme
alone, as these authors conjectured.
\end{abstract}


\section{Introduction}
\label{sec:intro}

Subdivision is one of the most widely used constructive paradigms in
geometric modelling. Starting from a control polygon, a subdivision scheme
repeatedly refines the data by inserting new points according to a fixed
local rule, and under suitable conditions the polygons converge to a smooth
limit curve. When the scheme is interpolatory, the original points are
retained at every level, so the limit curve passes through all of the input
data. The classical examples are linear, in that the inserted point is an affine
combination of nearby control points, as in the four-point scheme of Dyn,
Levin, and Gregory~\cite{dyn1987four} and, more generally, the
Deslauriers--Dubuc family~\cite{deslauriers1989symmetric}. For such linear
schemes there is by now a mature and essentially complete theory, surveyed
in~\cite{dyn2002acta}, which characterises convergence and smoothness
through the machinery of derived schemes and joint spectral
quantities~\cite{daubechies1992two,rioul1992}.
In particular, the six-point Deslauriers--Dubuc scheme, which plays a
central role in this paper and which we abbreviate DD6, has critical
H\"older exponent approximately $2.830$, while the four-point scheme
stops one full class lower.

Linear rules, however, are blind to geometry. An affine combination knows
nothing about circles, arc length, or the surface on which the data may
live, and this observation has motivated a substantial body of work on
geometric subdivision, in which the insertion is defined by a geometric
construction rather than by fixed coefficients. Well-known examples include
the circle-preserving variant of the four-point scheme by Sabin and
Dodgson~\cite{sabin2005circle}, the tangent-continuity analysis of
geometric planar rules by Dyn and Hormann~\cite{dyn2012geometric}, and the
clothoid-based Hermite subdivision of Reif and
Weinmann~\cite{reif2021clothoid}, in which segments of linearly varying
curvature are fitted to point--tangent data. More recently, intrinsic
geodesic-curvature constructions have carried these ideas to spherical and
hyperbolic
geometry~\cite{ahanchaou2022hyperbolic,bellaihou2024spherical,ahanchaou2026class}.
A related and equally active line of research asks how an arbitrary linear
scheme can be transferred to a Riemannian manifold, typically through the
log-exponential lift, and what smoothness survives the
transfer~\cite{wallner2005convergence,grohs2010general,wallner2011riemannian,wallner2014interpolatory}.
For point-only geometric rules in the plane, the sharpest available
regularity framework is that of Ewald, Reif, and
Sabin~\cite{ewald2015holder}, who axiomatise similarity-equivariant
local schemes, measure the departure of a polygon from uniform straight
data by a relative distortion functional, and derive convergence and
H\"older regularity from the decay of that functional under refinement,
a property they call straightening. That framework governs the second
half of this paper.

The scheme we study in this paper sits within the geometric family, but it
defines the insertion in a different way. Rather than fitting a particular
curve type, we let a curvature-profile model define the inserted point.
The idea is simple to state. For each edge of the polygon, we first
estimate the geodesic curvature at the two endpoints using circumscribed
circles. We then prescribe how the curvature should vary along the yet
unknown segment joining the endpoints. This prescription is the curvature
profile. Given the profile, the segment itself is reconstructed by
integrating the intrinsic Frenet equations, with the initial direction and
the total arc length determined by a two-parameter shooting problem which
forces the segment to arrive at the far endpoint. Finally, the new point is
taken at half arc length along the reconstructed segment. Every ingredient
of this construction is intrinsic, so the same rule operates verbatim in
the Euclidean plane and on the unit sphere. We refer to the scheme as
point-only because, unlike Hermite-type constructions, it consumes and
produces nothing but points.

Why should one expect such a rule to behave well? The motivation comes from
fairness. The curvature-variation energy $\int (\kappag')^2\, ds$ is the
classical fairness measure of Moreton and S\'equin~\cite{moreton1992}, and
a curve is visually pleasing precisely when its curvature varies slowly.
Fairness functionals of this kind have a long history in geometric design. In particular, they underpin the partial-differential-equation approach, in
which shapes are generated as solutions of elliptic boundary-value
problems~\cite{UgailBW99a,UgailBW99b,KubiesaUW04,GonzalezUgail08}. That
programme extends from the variational treatment of harmonic and
biharmonic B\'ezier
surfaces~\cite{monterde2004harmonic,monterde2006general} through parametric
aircraft geometry~\cite{AthanasopoulosUgail09} to mesh approximation and
facial parameterisation~\cite{ShengUgail10,ShengUgail11}, and it supplies
the design setting in which curvature-based fairness measures are used in
practice. The profile we ultimately adopt is drawn from exactly this variational source, for it linearises the Euler--Lagrange equation of the
fairness energy about the constant-curvature state, a derivation we take
from the variational analysis of the six-point scheme
in~\cite{ugail2026fair}, which itself supersedes an earlier preprint of the
same programme~\cite{ugail2026biharmonic} and stands in the lineage of
variational subdivision initiated by
Kobbelt~\cite{kobbelt1996variational}.

What we find when we analyse the resulting rule is, we believe, both
surprising and instructive, and it gives the paper its narrative. The raw
version of the scheme, in which the circumscribed-circle estimates are used
directly, has excellent geometry. It reproduces geodesics exactly, it
reproduces circles, and on the sphere small circles, exactly, and at coarse
refinement levels it is measurably fairer than the six-point comparator.
Under continued refinement, however, its fairness deteriorates, with the
measured energy growing geometrically. It would be easy to dismiss this as
a numerical artefact, but it is nothing of the sort. By expanding the
scheme exactly on perturbed straight-line data, we prove that the raw rule
linearises precisely to the four-point scheme, whose critical H\"older
exponent is one full class below that of the six-point scheme. In other
words, the deterioration is structural, and no amount of tuning the
curvature profile can remove it. The same computation, however, also
reveals the cure. If the curvature estimates are passed through the
three-point prefilter with weights $(-3, 22, -3)/16$ before being used,
then the linearisation becomes exactly the six-point Deslauriers--Dubuc
scheme, and we prove that this filter is the unique symmetric
constant-preserving three-point filter with that property. Because its
coefficients sum to one, it maps constant data to themselves, and the exact
circle reproduction of the raw rule survives filtering unchanged. The
final operator of the paper combines this filter with the variationally
adapted profile, and on the reported benchmarks it is fairer than the
six-point comparator at every tested level, in the plane and on the sphere
alike.

The linearisation theorems settle the structure of the rule, and the
second half of the paper settles its analysis, by a route that is
instructive in itself. One might expect the framework
of~\cite{ewald2015holder} to apply to the filtered scheme after routine
verification, since it is exactly the kind of operator that framework
was built for. We prove that it does not, and that the repair required
is exactly the enhancement conjectured at the end
of~\cite{ewald2015holder}. In more detail, we first prove that the
filtered rule restricted to collinear data of arbitrary spacing is
exactly midpoint subdivision, and that its full planar linearisation at
the uniform straight chain splits into midpoint subdivision acting on
tangential perturbations and DD6 acting on normal perturbations, with
vanishing cross terms. This establishes that the scheme satisfies every
structural axiom of the framework, while simultaneously placing it
outside the reach of the asymptotic theorems there, whose hypotheses
require either that the tangential and normal companions coincide or
that both be smooth schemes. We then prove that the straightening
condition, which underlies every theorem of the framework including
plain convergence, fails for the scheme on entire classes of data. The
decisive example is a nonuniformly sampled circle. The scheme reproduces
the circle exactly and inserts at half arc length, so the nonuniform arc
pattern is preserved forever, the relative distortion stays bounded away
from zero at every level, and yet the refined polygons converge to the
exact circle. The failure is confined to, and fully explained by, the
tangential block of the linearisation, whose second-difference decay
factor equals the edge-shrink rate exactly, so that tangential
nonuniformity is neither amplified nor removed.

The constructive part of the analysis replaces relative distortion by
its normal part. We prove a two-step contraction of the normal
distortion at the linearised level with the exact rational constant
$2697/4096$, define an admissible class of polygons that the operator
maps into itself with edges halving to third order and the safeguard of
the construction never activating, and prove that on this class the
tangential discrepancy between the scheme and DD6 is a pure change of
parameter of second order, while the normal discrepancy is of fourth
order. On arc-length uniform data both components are of fourth order,
which identifies the tangential term exactly. The main theorem then
states that on sufficiently dense admissible data the refined polygons
converge to a regular interpolatory limit curve whose unit tangent and
curvature are continuous functions of arc length, so that the limit is
$G^2$ and its arc-length parametrisation is $C^2$. The circle example
shows that the dyadic parametrisation of the same limit need not be
$C^1$, so the geometric formulation is forced rather than chosen. For
this scheme, the theorem realises the behaviour conjectured in the
conclusion of~\cite{ewald2015holder}, namely that regularity is governed
by the normal companion alone while the tangential companion is
irrelevant, and we discuss in Section~\ref{sec:discussion} what would be
needed to prove the conjecture in its general form.

Before proceeding, we briefly mention two adjacent research lines which
supply further context. Beyond Riemannian surfaces, interpolatory
subdivision has recently been extended to the sub-Riemannian setting of the
Heisenberg group, where the central coordinate exhibits a structural
smoothness loss~\cite{ugail2026heisenberg}. A complementary direction
replaces fixed insertion rules with learned components while retaining
structural guarantees. A shared neural tension operator predicts per-edge
insertion angles across the constant-curvature geometries, with safety
bounds enforced
architecturally~\cite{ugail2026neural,ugail2026adaptive}, and a
proximity-preserving neural refinement rule constrains a learned correction
to a quadratic envelope of a classical stencil~\cite{ugail2026pns}. The
scheme of this paper derives its per-edge adaptivity from a variational
model rather than from training data, and the envelope perspective of the
latter work is close in spirit to the admissibility safeguard we introduce
in Section~\ref{sec:scheme}.

The remainder of the paper is organised as follows. In
Section~\ref{sec:scheme} we define the scheme, beginning with an
overview of the final operator and then developing each of its
ingredients in turn. Section~\ref{sec:solvable} establishes that the
underlying shooting problem is locally well posed and that the scheme
reproduces geodesics and constant-curvature curves.
Section~\ref{sec:linearisation} contains the exact planar linearisation
of the raw and the filtered rules, the uniqueness of the prefilter, the
exact restriction to lines, and the full vector linearisation.
Section~\ref{sec:distortion} presents the framework of relative
distortion, proves membership, and proves the failure of straightening.
Section~\ref{sec:normal} develops the normal distortion, the invariant
admissible class, and the contraction estimates.
Section~\ref{sec:proximity} proves the proximity decomposition modulo
reparametrisation, and Section~\ref{sec:main} assembles the main
convergence and regularity theorem. Section~\ref{sec:numerics} reports
the numerical study, in the plane and on the sphere, and
Section~\ref{sec:discussion} concludes with a discussion of the open
problems. The analytical results we prove are planar, and the spherical
results provide strong experimental support for the same mechanisms.
Every symbolic claim in the paper is verified by exact computation in
the accompanying notebooks, which also generate every figure and
numerical value we report.

\section{The Curvature-Driven Subdivision Scheme}
\label{sec:scheme}

Throughout the paper, $M$ denotes either the Euclidean plane or the unit
sphere $\Stwo \subset \R^3$, so that the constant curvature is $K = 0$ or
$K = 1$, respectively. Control polygons are closed sequences
$p = (p_j)_{j \in \Z_n}$ of points of $M$, and $h$ denotes the maximum
geodesic edge length. Signed geodesic curvature follows one orientation convention in both geometries, whereby polygons are
traversed counterclockwise, and
the sign is taken against the surface orientation.

It is helpful to state the final operator in words before developing its
parts. The scheme, which we denote $\Tf$, maps the polygon $p$ to the
refined polygon with $\Tf(p)_{2j} = p_j$, so that the scheme is
interpolatory, and with $\Tf(p)_{2j+1}$ constructed on the edge
$(p_j, p_{j+1})$ in four steps. First, raw curvature estimates are computed
at the stencil vertices from circumscribed circles. Second, these estimates
are passed through the curvature prefilter derived in
Section~\ref{sec:linearisation}, giving the edge data
$(\hat\kappa_j, \hat\kappa_{j+1})$. Third, the adapted curvature profile
determined by these data is prescribed along the unknown segment, and the
segment is reconstructed by integrating the intrinsic Frenet equations,
shooting on the initial angle and the arc length so that the segment
connects $p_j$ to $p_{j+1}$. Fourth, the inserted point is taken at half
arc length along this segment. In the rare event that the shooting
iteration fails to converge within the admissible neighbourhood of
Theorem~\ref{thm:solvable}, the rule reverts to the raw curvature data on that edge. We return to this safeguard, and quantify how rarely it is
needed, in Section~\ref{sec:numerics}, and Theorem~\ref{thm:invariance}
will show that on the admissible data class of Section~\ref{sec:normal}
it never activates at all. The raw rule, which we denote
$\Traw$, is the same construction without the prefilter, and it serves
throughout the paper as the baseline whose exact analysis motivates the
final design. We now describe each ingredient in detail.

\subsection{Estimating Curvature from Circumscribed Circles}
\label{sec:estimator}

The first ingredient is a local estimate of the geodesic curvature at each
vertex, and the natural choice is the curvature of the circumscribed
circle through the vertex and its two neighbours. In the plane, this is the
well-known signed Menger curvature,
\begin{equation}
  \kappa_j^{M}
  =
  \frac{2\det\!\bigl[p_j - p_{j-1},\ p_{j+1} - p_j\bigr]}
       {\norm{p_j - p_{j-1}}\,\norm{p_{j+1} - p_j}\,\norm{p_{j+1} - p_{j-1}}}.
  \label{eq:menger}
\end{equation}
On the sphere, the analogous construction is equally simple. The circle
through three points of $\Stwo$ is the intersection of the sphere with the
plane spanned by those points. Hence, if $n$ is the unit normal of that
plane, oriented by the traversal, and $c = n \cdot p_j$ is its offset, the
estimate is
\begin{equation}
  \kappa_j^{M} = \frac{c}{\sqrt{1 - c^2}} = \cot \rho ,
  \label{eq:sphcirc}
\end{equation}
where $\rho$ is the intrinsic radius of the circle and the sign is carried
by the orientation of $n$. A point worth emphasising, since the
reproduction results of Section~\ref{sec:solvable} depend on it, is that
this estimator is exact on curves of constant geodesic curvature. If the
three stencil points lie on a geodesic or on a circle, the estimate equals
the true curvature with no error at all.

\subsection{Curvature Profiles}
\label{sec:profiles}

Having estimated the curvatures $\kappa_0$ and $\kappa_1$ at the two ends
of an edge, we must next decide how the curvature should vary along the
segment which will join them. This decision is the curvature profile, and
it is where the variational model enters the construction.

Recall that the fairness energy $\int (\kappag')^2\, ds$ measures the
variation of curvature along a curve. On a surface of constant curvature
$K$, the Euler--Lagrange equation of this energy is the nonlinear
fourth-order equation
$2\kappag^{(4)} + 2(\kappag^2 + K)\kappag'' - \kappag(\kappag')^2 = 0$, a
derivation of which may be found in~\cite{ugail2026fair}. Every constant
is an interior stationary solution of this equation, which reflects the
fact that circles and geodesics are the fairest curves of all. This
suggests linearising the equation about the constant state. Writing
$\bar\kappa = \tfrac12(\kappa_0 + \kappa_1)$ for the mean of the edge data
and $\kappag = \bar\kappa + \eta$, the curvature-sensitive factor of the
linearised operator gives what we call the adapted profile, in which the function $\eta$ solves
\begin{equation}
  \eta'' = -\bigl(\bar\kappa^2 + K\bigr)\eta, \qquad
  \eta(0) = \kappa_0 - \bar\kappa, \qquad \eta(L) = \kappa_1 - \bar\kappa
  \label{eq:adapted}
\end{equation}
on the unknown length interval $[0, L]$. Two consequences follow directly
from this choice, and both will matter later. First, when
$\kappa_0 = \kappa_1$ the boundary data of $\eta$ vanish, so the profile is exactly constant, and this is what drives the exact circle reproduction of
Section~\ref{sec:solvable}. Second, near a geodesic the mean $\bar\kappa$
is small, and the profile reduces to the kernel of
$\partial_{ss} + K$.

For the analysis and the experiments, it is useful to keep two alternative
profiles alongside the adapted one, as ablations which explain why the
final choice is the right one. The clothoid profile interpolates the data
linearly in arc length, so that in the plane the reconstructed segments are
Euler spirals, exactly as in the clothoid-based Hermite subdivision
of~\cite{reif2021clothoid}, and it agrees with the adapted profile whenever the
two data coincide or the mean curvature is small. The sensitive profile
solves $\kappag'' = -K\kappag$ with the same boundary data. As we shall see in
Section~\ref{sec:solvable}, it cannot reproduce nongeodesic circles when
$K \neq 0$, and this obstruction is precisely why it is not the final
choice.

\subsection{Segment Reconstruction by Geodesic Shooting}
\label{sec:shooting}

With the profile fixed, we can now reconstruct the segment itself. The
candidate segment $\gamma_{\phi, L}$ issues from $p_j$ with its initial
tangent making an angle $\phi$ with the geodesic direction towards
$p_{j+1}$, and it evolves according to the intrinsic Frenet equations. On
the unit sphere, writing $x = \gamma$ and $T = \gamma'$, these read
\begin{equation}
  x' = T, \qquad T' = \kappag(s)\,(x \times T) - x ,
  \label{eq:frenet}
\end{equation}
where the final term accounts for the normal curvature of the embedding. In the plane, the system reduces to the familiar $\theta' = \kappag$ with
$\gamma' = (\cos\theta, \sin\theta)$. Of course, an arbitrary choice of the
angle $\phi$ and the length $L$ will not deliver the segment to the far
endpoint. Writing $F(\phi, L)$ for the residual of
$\gamma_{\phi, L}(L)$ against $p_{j+1}$, expressed in a fixed tangent basis
at $p_{j+1}$, the two unknowns must therefore solve the shooting problem
\begin{equation}
  F(\phi, L) = 0 .
  \label{eq:shoot}
\end{equation}
This is a two-by-two nonlinear system, and we solve it by Newton's method
started from the geodesic initial guess, namely $\phi = 0$ and $L$ equal to
the geodesic distance $d(p_j, p_{j+1})$. Note that the profile
\eqref{eq:adapted} itself depends on $L$, so the system is genuinely
nonlinear in both unknowns. Once the solution $(\phi^\ast, L^\ast)$ is
found, the inserted point is simply
$\gamma_{\phi^\ast, L^\ast}(\tfrac12 L^\ast)$, the half-arc-length point of
the segment.

\subsection{The Final Operator and its Support}
\label{sec:finalop}

Assembling the pieces, the final operator $\Tf$ applies the prefilter of
Section~\ref{sec:linearisation} to the raw circumscribed-circle estimates,
prescribes the adapted profile with the filtered data, reconstructs the
segment by shooting, and inserts the half-arc-length point, reverting to
raw data on an edge only when the safeguard described above is triggered.

It is worth recording the support of the rule precisely, since it makes
the comparison with the linear schemes exact. For the edge
$(p_j, p_{j+1})$, the raw endpoint estimates involve the two circumscribed
circles through $p_{j-1}, p_j, p_{j+1}$ and through $p_j, p_{j+1}, p_{j+2}$. Hence the raw rule has four-point
support,
matching the four-point scheme. The filtered pair involves $p_{j-2}$
through $p_{j+3}$, so the filtered rule has six-point support, matching the
six-point scheme. As we shall prove in Section~\ref{sec:linearisation},
these coincidences of support are not accidental, for the raw rule
linearises
to the four-point scheme and the filtered rule to the six-point scheme.
Finally, we note that the rule is interpolatory and intrinsic, and that it
is equivariant under the isometries of $M$, because every ingredient of
the construction, from circumscribed circles through Frenet integration to
arc length, is itself isometry invariant.

\section{Local Solvability and Reproduction}
\label{sec:solvable}

Before any analysis of the scheme's behaviour can begin, we must first
confirm that the scheme is well defined at all. The shooting problem
\eqref{eq:shoot} prescribes the curvature profile and both endpoint
positions, which for fixed $\phi$ and $L$ over-determines the segment. Solvability therefore rests entirely on the two free parameters. The
natural tool for such a question is the implicit function theorem applied
at the geodesic solution, and it leads to the following local statement.

\begin{theorem}[Local solvability of the shooting problem]
\label{thm:solvable}
Let $p_j, p_{j+1} \in M$ with geodesic distance $d > 0$, and on the sphere
let the pair be nonconjugate, so that $d < \pi$. Then there is a
neighbourhood of zero in the curvature data such that for all
$(\kappa_0, \kappa_1)$ in that neighbourhood the shooting problem
\eqref{eq:shoot} has a unique solution $(\phi^\ast, L^\ast)$ in a
neighbourhood of $(0, d)$, depending smoothly on the data.
\end{theorem}

\begin{proof}
At zero curvature data every profile vanishes identically, and
$(\phi, L) = (0, d)$ solves \eqref{eq:shoot} with the geodesic segment. To
apply the implicit function theorem we must show that the Jacobian
$D_{(\phi, L)} F$ is invertible at this solution, and this is where the
geometry enters. The partial derivative of $F$ with respect to $L$ is the
terminal unit tangent of the geodesic. The partial derivative with respect
to $\phi$ is the endpoint value of the transverse Jacobi field generated by
rotating the initial tangent. This field is orthogonal to the terminal
tangent, and its magnitude is $d$ in the plane and $\sin d$ on the unit
sphere. The determinant of $D_{(\phi, L)} F$ is therefore proportional to
this transverse Jacobi factor, which is nonzero for $d > 0$ in the plane
and for nonconjugate pairs on the sphere. Since the residual depends
smoothly on the unknowns and the data, because the profile
\eqref{eq:adapted} and its ablation variants are smooth in
$(\kappa_0, \kappa_1, L)$ away from their resonances, which are excluded
near $(0, d)$, the implicit function theorem applies.
\end{proof}

The Jacobi proportionality at the heart of this proof can be checked
directly, and we do so in the accompanying notebook, where the finite-difference
determinant divided by $\sin d$ equals $-1.0000$ to
within $1.1 \times 10^{-4}$ across edge lengths from $0.2$ to $2.2$ on the
sphere. The theorem establishes that the construction is well defined near
geodesic data, and in practice the solution basin turns out to be far
larger than the guaranteed neighbourhood. Across the fairness and stress
experiments of Section~\ref{sec:numerics}, comprising $6451$ planar and
$593$ spherical insertions on the smooth benchmarks together with $1512$ insertions on the stress
example, every retained Newton solve converged from the geodesic initial
guess, with a median of one planar and three spherical updates, maxima of
three and five, final residuals below $4.4 \times 10^{-14}$ and
$9.7 \times 10^{-13}$, positive lengths throughout, and Jacobian condition
numbers below $250$ and $33$. One exception, on a noisy stress edge where
the Newton iteration escaped the local branch and was intercepted by the
admissibility safeguard, is documented in Section~\ref{sec:numerics}.

Having settled solvability, we can now ask what the scheme does on the
simplest possible data, namely samples of geodesics and circles. The
answer follows from two facts we have already noted, namely that the
estimator is exact on constant-curvature curves and that the adapted
profile is exactly constant on equal data.

\begin{proposition}[Local geodesic and constant-curvature reproduction]
\label{prop:repro}
Let the stencil points lie in order on a geodesic, or on a complete curve
$C$ of constant geodesic curvature $\kappa_c \neq 0$, that is a circle in
the plane or a small circle on the sphere, and let the sampling be dense
enough that consecutive points are nonconjugate and the data
$(\kappa_c, \kappa_c)$ lie in the solvability neighbourhood of
Theorem~\ref{thm:solvable} for each edge, with the arc of $C$ between the
endpoints contained in the corresponding local solution branch. Then the
adapted rule, raw or filtered, inserts the geodesic midpoint in the first
case and a point of $C$ in the second. The sensitive ablation profile
admits no nonzero constant solution for $K \neq 0$ and therefore does not
reproduce nongeodesic circles on the sphere.
\end{proposition}

\begin{proof}
Consider first the geodesic case. Here both circumscribed circles
degenerate and all estimates vanish, so the profile vanishes, and the
unique local shooting solution is the geodesic segment, whose
half-arc-length point is the geodesic midpoint. Now consider
constant-curvature data. Both estimates equal $\kappa_c$ exactly, and the
prefilter maps the constant sequence to itself because its coefficients
sum to one, so the edge data are $(\kappa_c, \kappa_c)$ in both variants.
The adapted profile is then exactly constant, since
$\bar\kappa = \kappa_c$ and the boundary data of $\eta$ vanish. A
constant-curvature segment of curvature $\kappa_c$ joining two points of
$C$ at the correct orientation is an arc of $C$. It satisfies the endpoint
condition, and by hypothesis it is the unique local shooting solution, so
the half-arc-length point lies on $C$. Finally, for the sensitive
obstruction, substituting a constant into $\kappag'' = -K\kappag$ forces
$K\kappag = 0$, which has no nonzero constant solution when $K \neq 0$.
\end{proof}

These reproduction properties are confirmed numerically to the precision
one would hope for. Planar circle reproduction holds to
$1.1 \times 10^{-16}$ for the raw and the filtered rule alike, spherical
small-circle reproduction holds to $7.0 \times 10^{-11}$, which is the
tolerance of the segment integrator, and the sensitive ablation incurs a
spherical circle error of $4.8 \times 10^{-4}$ at the tested density,
exactly as the proposition predicts. For comparison, the six-point scheme
does not reproduce circles at all. At the same planar density its radial
error is $2.8 \times 10^{-4}$. Exact constant-curvature reproduction is
thus a genuine property of the final operator, and one not possessed by
the stationary affine Deslauriers--Dubuc schemes we compare against.

\section{The Exact Planar Linearisation}
\label{sec:linearisation}

We now come to the central analytical results of the paper. The question we
wish to answer is why the raw rule, despite its excellent geometry, loses
fairness under refinement, and the standard route to such questions runs
through the scheme's linearisation on perturbed regular
data~\cite{ewald2015holder,grohs2010general}. The asymptotic behaviour
of a
nonlinear subdivision scheme is governed by the linear scheme it reduces to
near regular configurations. What is remarkable about the present rule is
that this linearisation can be computed exactly, and that the answer turns
out to be a scheme everyone knows.

\subsection{Setting up the Expansion}
\label{sec:expansion}

We fix the perturbation framework first. Consider uniformly sampled
straight-line data perturbed in the normal direction,
$p_j = (jh,\ \varepsilon y_j)$, and expand every ingredient of the planar
rule to first order in $\varepsilon$. All displayed identities below
relate
the first-order coefficients. Three expansions are needed. The raw Menger
estimate satisfies $\kappa_j^{M} = \varepsilon k_j + O(\varepsilon^2)$
with
\begin{equation}
  k_j = \frac{y_{j-1} - 2y_j + y_{j+1}}{h^2},
  \label{eq:menger_lin}
\end{equation}
which is simply the scaled second difference, as one would expect of a
curvature estimate. The segment length satisfies
$L = h + O(\varepsilon^2)$. Finally, the inserted point satisfies
$\Traw(p)_{2j+1} =
\bigl((j + \nicefrac12)h,\ \varepsilon\, y_{j+\nicefrac12}\bigr) +
O(\varepsilon^2)$, and integrating the small-slope Frenet equations with
any profile that interpolates the edge data linearly to first order, and
then enforcing the endpoint condition, yields the first-order midpoint
coefficient
\begin{equation}
  y_{j + \nicefrac12}
  = \frac{y_j + y_{j+1}}{2}
  - \frac{h^2}{16}\,\bigl(k_j^{\mathrm{edge}} + k_{j+1}^{\mathrm{edge}}\bigr),
  \label{eq:midpoint_lin}
\end{equation}
where $(k_j^{\mathrm{edge}}, k_{j+1}^{\mathrm{edge}})$ are the first-order
coefficients of the curvature data actually supplied to the edge. Formula
\eqref{eq:midpoint_lin} is the entire linearised scheme. Everything now
depends on what those data are, and this is precisely the lever that the
prefilter will pull.

\subsection{The Raw Rule Linearises to the Four-Point Scheme}
\label{sec:rawlin}

Suppose first that the raw estimates are used directly. Substituting
\eqref{eq:menger_lin} into \eqref{eq:midpoint_lin} gives the following
result, which explains the raw rule's behaviour completely.

\begin{theorem}[Linearisation of the raw rule]
\label{thm:dgl}
With raw data \eqref{eq:menger_lin}, the linearisation of the planar raw
rule on perturbed uniform straight-line data is exactly the four-point
scheme of Dyn, Levin, and Gregory,
\begin{equation}
  y_{j+\nicefrac12}
  = \frac{-y_{j-1} + 9y_j + 9y_{j+1} - y_{j+2}}{16}.
  \label{eq:dgl}
\end{equation}
\end{theorem}

\begin{proof}
Substituting \eqref{eq:menger_lin} into \eqref{eq:midpoint_lin} gives
\[
  \frac{y_j + y_{j+1}}{2}
  - \frac{(y_{j-1} - 2y_j + y_{j+1}) + (y_j - 2y_{j+1} + y_{j+2})}{16}
  = \frac{-y_{j-1} + 9y_j + 9y_{j+1} - y_{j+2}}{16}.
\]
Both \eqref{eq:midpoint_lin} and this substitution are verified by exact
symbolic computation in the accompanying notebook.
\end{proof}

\begin{corollary}[The obstruction]
\label{cor:barrier}
The critical H\"older exponent of the four-point scheme is two, one full
regularity class below the six-point value of approximately $2.830$. Both values follow from the classical regularity machinery for stationary
schemes~\cite{daubechies1992two,rioul1992}. Since the raw nonlinear rule
has the four-point scheme as its normal linearisation, a six-point-type
$C^2$ regularity conclusion cannot follow from its linearised core.
Moreover, the geometric growth of its measured curvature-variation energy
under refinement, whose ratio will be seen in
Section~\ref{sec:numerics} to stabilise at the four-point signature value,
is exactly the manifestation one would expect of this obstruction.
\end{corollary}

The corollary deserves a moment of reflection, because it recasts the raw
rule's failure in a useful light. The failure is not a defect of the
implementation, nor of the curvature estimator, nor of any particular profile. It is a structural property of the construction, visible already
at first order, and consequently no tuning of the profile can remove it.
The value of an exact linearisation is that it tells us this with
certainty, and, as we show next, it simultaneously tells us what the
repair must be.

\subsection{The Curvature Prefilter}
\label{sec:filter}

The midpoint formula \eqref{eq:midpoint_lin} shows that the linearised
scheme is determined entirely by the curvature data supplied to the edge.
This raises the natural question of whether there is a preprocessing of
the raw estimates whose effect, at the linearised level, is to turn the
four-point mask into the six-point one. The answer is affirmative, and
uniquely so.

\begin{theorem}[The DD6-compatible curvature prefilter]
\label{thm:filter}
Let the edge data be the filtered estimates
\begin{equation}
  \hat\kappa_j =
  \frac{-3\kappa_{j-1}^{M} + 22\kappa_j^{M} - 3\kappa_{j+1}^{M}}{16}.
  \label{eq:filter}
\end{equation}
Then the linearisation of the planar filtered rule on perturbed uniform
straight-line data is exactly the six-point Deslauriers--Dubuc scheme,
\begin{equation}
  y_{j+\nicefrac12}
  = \frac{3y_{j-2} - 25y_{j-1} + 150y_j + 150y_{j+1} - 25y_{j+2} +
  3y_{j+3}}{256},
  \label{eq:dd6lin}
\end{equation}
and \eqref{eq:filter} is the unique filter with this property among
symmetric, constant-preserving three-point curvature filters. Because its
coefficients sum to one, it maps constant curvature data to themselves, so
the reproduction results of Proposition~\ref{prop:repro} survive filtering
unchanged.
\end{theorem}

\begin{proof}
Substituting \eqref{eq:menger_lin} into \eqref{eq:filter} gives the
first-order coefficient
\[
  h^2 \hat k_j
  = \frac{-3y_{j-2} + 28y_{j-1} - 50y_j + 28y_{j+1} - 3y_{j+2}}{16},
\]
so that $h^2(\hat k_j + \hat k_{j+1}) = (-3y_{j-2} + 25y_{j-1} - 22y_j
- 22y_{j+1} + 25y_{j+2} - 3y_{j+3})/16$, and \eqref{eq:midpoint_lin} then
yields \eqref{eq:dd6lin}. Both steps are verified symbolically in the
accompanying notebook. For the uniqueness claim, observe that a
symmetric constant-preserving three-point filter has weights
$(a, 1 - 2a, a)$, and inserting it into \eqref{eq:midpoint_lin} produces a
six-point mask whose coefficients are affine functions of $a$. Matching
any single outer coefficient of the Deslauriers--Dubuc mask, for instance
$3/256$ at $y_{j-2}$, forces $a = -3/16$, and one checks that this value
matches all six coefficients simultaneously.
\end{proof}

The theorem cancels the four-point barrier in the planar linearisation and
replaces it with the six-point linearisation. Converting this exact
linearised statement into nonlinear convergence and a regularity class is
the business of Sections~\ref{sec:distortion} to~\ref{sec:main}. What
remains to be checked is that the filter
is compatible with the adapted profile of the final operator, and this is
the content of the following corollary.

\begin{corollary}[The filtered adapted operator]
\label{cor:final}
The final operator of Section~\ref{sec:scheme} possesses both intended
structural properties at once. On the one hand, the prefilter preserves
constant curvature, so the adapted profile continues to reproduce circles
exactly after filtering. On the other hand, near a straight line the mean
$\bar\kappa$ is of first order in $\varepsilon$, so the coefficient
$\bar\kappa^2 + K$ of the adapted profile differs from $K$ only at second
order, and the profile agrees with linear interpolation of the edge data
to first order. Consequently, the filtered adapted operator inherits the
six-point linearisation of Theorem~\ref{thm:filter} verbatim.
\end{corollary}

The first-order agreement asserted by the corollary can also be observed directly. On randomly perturbed line data, the adapted and clothoid
insertions agree to below $10^{-11}$ at perturbation amplitudes $10^{-2}$
and $10^{-3}$, which is the solver tolerance, so the observed agreement is
in fact stronger than the first-order statement requires.

The filter itself merits a further remark. Its side coefficients are
negative, so it is a sharpening operator rather than an average, and its
role is to cancel, at the linearised level, the discrepancy between the
three-point circumscribed-circle estimate and the curvature data that the
six-point insertion implicitly encodes. Sharpening, of course, can amplify
noise, and in Section~\ref{sec:numerics} we document one insertion on a
noisy stress example in which precisely this happened and the admissibility
safeguard of Section~\ref{sec:scheme} intervened. On the sphere, we apply
\eqref{eq:filter} to the estimates \eqref{eq:sphcirc} verbatim. A spherical
analogue of Theorem~\ref{thm:filter} is not proved here, and the spherical
results below are accordingly reported as experiments.

\subsection{Restriction to Lines and the Full Vector Linearisation}
\label{sec:fulllin}

Theorems~\ref{thm:dgl} and~\ref{thm:filter} linearise the rule under
normal perturbations of uniform straight data. The regularity theory to
come requires the full vector linearisation, and the natural starting
point is the observation that the scheme can be restricted to a line
exactly, with no perturbation expansion at all.

\begin{theorem}[Exact restriction to lines]
\label{thm:restriction}
Let the six stencil points $p_{j-2}, \ldots, p_{j+3}$ lie in order on a
straight line, with arbitrary positive spacings. Then the filtered and the
raw insertion on the edge $(p_j, p_{j+1})$ both equal the Euclidean
midpoint $\tfrac12(p_j + p_{j+1})$. Consequently $\Tf$ restricted to
collinear polygons is exactly midpoint subdivision, for every spacing
pattern.
\end{theorem}

\begin{proof}
On collinear data every determinant in \eqref{eq:menger} vanishes, so all
raw Menger estimates are zero, and the prefilter \eqref{eq:filter} maps
the zero sequence to itself. The edge data are therefore $(0, 0)$, the
mean $\bar\kappa$ vanishes, and the boundary data of \eqref{eq:adapted}
vanish, so the adapted profile is identically zero. The reconstructed
segment with zero curvature profile is a straight segment, the unique
local solution of the shooting problem is the chord itself with $\phi = 0$
and $L$ equal to the chord length, and its half-arc-length point is the
Euclidean midpoint. No uniformity of the spacing was used at any step. The
identity is confirmed to $4.4 \times 10^{-16}$ on nonuniform collinear
data in the accompanying notebook.
\end{proof}

The theorem pins down the tangential behaviour of the scheme. Writing
$\Smid$ for midpoint subdivision, whose limit functions are the piecewise
linear interpolants of the data and whose critical H\"older exponent is
exactly one, we now combine the restriction with the normal linearisation
of Theorem~\ref{thm:filter} to obtain the full derivative. We write
$\mathbf e$, as before, for the uniform straight
chain $p_j = (j, 0)$, and we consider perturbations
$p_j = (j + \varepsilon\tau_j,\ \varepsilon\nu_j)$ with tangential and
normal components $\tau_j$ and $\nu_j$.

\begin{theorem}[Full block linearisation]
\label{thm:block}
To first order in $\varepsilon$, the filtered insertion on the central
edge satisfies
\begin{equation}
  \Tf(p)_{2j+1}
  =
  \Bigl(j + \tfrac12 + \varepsilon\,\frac{\tau_j + \tau_{j+1}}{2},\;\;
  \varepsilon\,\frac{3\nu_{j-2} - 25\nu_{j-1} + 150\nu_j + 150\nu_{j+1}
  - 25\nu_{j+2} + 3\nu_{j+3}}{256}\Bigr)
  + O(\varepsilon^2),
  \label{eq:blocklin}
\end{equation}
so that the derivative of $\Tf$ at $\mathbf e$ is the block scheme
$\Smid \oplus \Sdd$, acting by midpoint subdivision on tangential
perturbations and by DD6 on normal perturbations, with vanishing cross
terms. The raw rule satisfies the same identity with $\Sdgl$ in place of
$\Sdd$.
\end{theorem}

\begin{proof}
We verify the four blocks of the derivative separately. For the
tangential response to tangential perturbations, observe that a purely
tangential perturbation keeps the data collinear, so by
Theorem~\ref{thm:restriction} the insertion is exactly the midpoint
$\tfrac12(p_j + p_{j+1})$ for every $\varepsilon$, whose tangential
coordinate is $j + \tfrac12 + \varepsilon(\tau_j + \tau_{j+1})/2$. The
tangential block is therefore exactly $\Smid$, and not merely to first
order. For the normal response to tangential perturbations, the same
restriction shows that collinear data produce a collinear insertion, so
the normal coordinate vanishes identically and this cross block is zero.
For the tangential response to normal perturbations, consider the
reflection $\nu \mapsto -\nu$ across the base line. The construction is
equivariant under this reflection, because every ingredient of the rule,
from circumscribed circles through the profile equation to arc length, is
isometry invariant, so the tangential coordinate of the insertion is an
even function of the normal perturbation and its first derivative
vanishes. Finally, the normal response to normal perturbations is
precisely the setting of Theorems~\ref{thm:dgl}
and~\ref{thm:filter}, which give the DD6 coefficient row for the
filtered rule and the four-point row for the raw rule. The blocks are
verified numerically in the accompanying notebook by Richardson-extrapolated
finite differences at $\mathbf e$, where the tangential
block matches the midpoint mask to $3.3 \times 10^{-12}$, the normal
block matches the DD6 mask to $1.0 \times 10^{-15}$, the raw normal block
matches the four-point mask to $1.1 \times 10^{-16}$, and both cross
blocks vanish to working precision.
\end{proof}

\section{Relative Distortion and the Failure of Straightening}
\label{sec:distortion}

The linearisation results place the filtered scheme squarely inside the
class of operators for which Ewald, Reif, and Sabin built their
regularity framework, and the natural next step is to run the scheme
through that framework. This section carries out the attempt and proves
that it fails, in a way that is structural rather than technical, and
the failure will dictate the design of the theory that follows.

\subsection{Relative Distortion and GLUE Schemes}
\label{sec:ers}

We now recall, in our notation, the framework of~\cite{ewald2015holder}.
From here through Section~\ref{sec:main} we work in the plane, where the
linearisation theorems of Section~\ref{sec:linearisation} hold exactly,
and we return to the sphere in Sections~\ref{sec:numerics}
and~\ref{sec:discussion}. The framework
concerns interpolatory-type local schemes acting on chains of points, and
its structural axioms name four properties. A scheme is a GLUE scheme if
it commutes with similarity transformations of the plane, if its two
refinement rules act locally on windows of $m + 1$ consecutive points, if
those rules are uniform in the window index and continuously
differentiable with H\"older continuous derivative in a neighbourhood of
the standard chain $\mathbf e$ of uniformly spaced collinear points, and
if equidistant collinear data are mapped to equidistant collinear data at
half the spacing, with an offset $\tau$ that fixes the parametric
alignment of the refined chain. Interpolatory schemes have $\tau = 0$.

The framework measures the departure of a chain $q$ of $n$ points from
uniform straight data by the relative distortion
\begin{equation}
  \kERS(q) = \frac{\abs{q}_2}{\abs{\Pi q}_1},
  \label{eq:distortion}
\end{equation}
in which $\abs{q}_2$ is the maximum Euclidean norm of the second
differences of the chain, $\Pi q$ is the best affine fit to the chain in
the window index, and $\abs{\Pi q}_1$ is the norm of the fitted increment.
The quantity is invariant under similarities, it vanishes exactly on
uniform collinear chains, and for a polygon $p$ we write $\kERS(p)$ for
the maximum of \eqref{eq:distortion} over all cyclic windows of $n$
consecutive points. A polygon sequence $p^{(\ell)} = T^\ell p$ is
straightened by the scheme if $\kERS(p^{(\ell)}) \to 0$, and it is
straightened at rate $\alpha$ if $2^{\ell\alpha}\kERS(p^{(\ell)})$ stays
bounded. Straightening is the entry condition for every theorem of the
framework. Convergence of the refined polygons is derived from
straightening, $C^1$ regularity from strong straightening, and
$C^{1,\alpha}$ regularity from straightening at rate $\alpha$, while the
workhorse that certifies straightening in examples is a finite-step
contraction estimate for \eqref{eq:distortion} on normalised chains.

The asymptotic part of the framework analyses the derivative of the
refinement rules at the standard chain. Similarity equivariance forces
this derivative to split into two linear schemes, a tangential companion
$A$ acting on perturbations along the chain and a normal companion $B$
acting on perturbations across it. Two results then connect the
companions to the regularity of the nonlinear scheme. If the data are
straightened and both $A$ and $B$ are almost $C^{1,\alpha}$ schemes, the
nonlinear scheme is almost $C^{1,\alpha}$. If in addition the scheme is
locally linear, meaning $A = B$, and the common companion is almost
$C^{2,\alpha}$, then the nonlinear scheme is almost $C^{2,\beta}$ for an
explicit $\beta$. These are the only routes to first-order and
second-order regularity in the framework. Its authors close by
conjecturing that the hypotheses can be weakened to conditions on
$B$ alone, so that the regularity of a geometric scheme would be governed
by its normal companion while the tangential companion is irrelevant. The
present paper proves that conjecture for the operator $\Tf$, in the
geometric formulation made precise in Section~\ref{sec:main}, and the
proof explains why nothing weaker would have sufficed.

Within the framework just recalled, the block linearisation of
Theorem~\ref{thm:block} has two immediate consequences. The first is that
$\Tf$ satisfies every structural axiom.

\begin{corollary}[Membership in the framework]
\label{cor:glue}
The operator $\Tf$ is a GLUE scheme with locality parameter $m = 5$,
spread $n = 11$, offset $\tau = 0$, and smooth rules near the standard
chain. Similarity equivariance holds because every ingredient of the
construction is invariant under rigid motions and scales correctly under
dilations, and it is confirmed to $7.7 \times 10^{-16}$ over random
similarities including reflections. Locality and uniformity hold by the
six-point support of the filtered rule. Smoothness of the rules in a
neighbourhood of $\mathbf e$ holds because the shooting solution depends
smoothly on the stencil by Theorem~\ref{thm:solvable}, so the rules are in fact infinitely
differentiable there and the H\"older requirement on the derivative is met
with exponent one. Equilinearity with $\tau = 0$ is the equal-spacing case
of Theorem~\ref{thm:restriction}.
\end{corollary}

The second consequence closes the direct route through the asymptotic
theorems of the framework.

\begin{corollary}[The companion pair]
\label{cor:companions}
The tangential companion of $\Tf$ is $A = \Smid$ and the normal companion
is $B = \Sdd$. Since the limit functions of $\Smid$ are the piecewise
linear interpolants of the data, the tangential companion is Lipschitz but
not $C^1$, so it is not an almost $C^{1,\alpha}$ scheme for any positive
$\alpha$, and the first-order theorem of~\cite{ewald2015holder} does not
apply to $\Tf$. Since $A \neq B$, the scheme is not locally linear, and
the second-order theorem of~\cite{ewald2015holder} does not apply either.
\end{corollary}

At this point one route through the framework remains open, namely the
direct certification of straightening by a contraction estimate for the
relative distortion, which would at least deliver convergence. The next
subsection proves that this route is closed as well, and that the closure is
not a technical artefact but a structural property of half-arc-length
insertion.

\subsection{The Failure of Straightening}
\label{sec:straightfail}

The relative distortion \eqref{eq:distortion} vanishes only on uniform
collinear chains, so it penalises two quite different departures from
straightness at once. A chain can be distorted because it bends, and it
can be distorted because its points are unevenly spaced along a perfectly
straight path. The scheme $\Tf$ removes the first kind of distortion and
deliberately preserves the second, and this section makes both halves of
that sentence precise.

\begin{proposition}[Neutrality of the tangential dynamics]
\label{prop:neutral}
Let $p$ be a collinear polygon with consecutive edge lengths $e_j > 0$.
Then $\Tf p$ is collinear, each edge splits into two equal halves, and the
multiset of ratios of consecutive edge lengths of $\Tf p$ consists of the
value one together with all the ratios of $p$. In particular the maximum
consecutive-edge ratio is invariant under refinement, and if $p$ is
nonuniform then $\kERS(T_{\mathrm{CV}}^{\ell}\, p)$ is bounded away from
zero uniformly in $\ell$, even though every polygon in the sequence lies
on the same straight line.
\end{proposition}

\begin{proof}
Collinearity and exact edge bisection are the content of
Theorem~\ref{thm:restriction}. The two children of the edge $e_j$ have
length $e_j/2$ each, so the consecutive ratios of the refined polygon are
$1$ within each parent edge and $e_{j+1}/e_j$ across each parent junction,
which proves the multiset statement and the invariance of the maximum
ratio $r = \max_j \max(e_{j+1}/e_j,\ e_j/e_{j+1})$. For the distortion
bound, fix a junction realising $r$ and consider at level $\ell$ the
window of $n$ points centred there. The second difference at the junction
has magnitude $\abs{e_{j+1} - e_j}/2^\ell$ while the fitted increment of
the window is at most the largest local edge $\max(e_j, e_{j+1})/2^\ell$,
so
\[
  \kERS\bigl(T_{\mathrm{CV}}^{\ell}\, p\bigr)
  \;\ge\;
  \frac{\abs{e_{j+1} - e_j}}{\max(e_j, e_{j+1})}
  \;=\; 1 - \frac{1}{r}
  \;>\; 0
\]
for every $\ell$, which is the claimed uniform lower bound.
\end{proof}

The proposition shows that straightening fails already on data whose
limit is a straight line, but one might object that collinear nonuniform
data are artificial. The next theorem removes the objection with the most
natural data the scheme possesses, and it is the decisive example of the
paper.

\begin{theorem}[The circle counterexample]
\label{thm:circle}
Let $p^{(0)}$ sample a circle of radius $R$ at angles with nonuniform
gaps $g_j > 0$, chosen dense enough that every edge lies in the local
solvability branch. Then every refined polygon
$p^{(\ell)} = T_{\mathrm{CV}}^{\ell}\, p^{(0)}$ samples the same circle,
at angular gaps $g_j / 2^{\ell}$ inherited from the initial pattern, and
the polygons converge to the exact circle in Hausdorff distance. If the
gaps are not all equal, there nevertheless exists $c > 0$, depending only
on the gap pattern, with
\[
  \kERS\bigl(p^{(\ell)}\bigr) \;\ge\; c
  \qquad\text{for every } \ell \ge 0 ,
\]
so the sequence is not straightened. Consequently the convergence,
$C^{1,\alpha}$, and $C^{2,\alpha}$ theorems of~\cite{ewald2015holder} all
have unsatisfiable hypotheses for $\Tf$, although the scheme converges on
these data to the smoothest curve there is.
\end{theorem}

\begin{proof}
By the reproduction results of Proposition~\ref{prop:repro}, on circle
data the estimator returns the exact curvature at every vertex, the
prefilter preserves the constant sequence, the adapted profile is exactly
constant, and the reconstructed segment on each edge is the shorter arc of
the circle itself. The half-arc-length point of a circular arc bisects its
central angle, so the refined polygon samples the same circle at the
halved gaps, and induction gives the gap pattern $g_j/2^\ell$ at level
$\ell$. Hausdorff convergence to the circle is then immediate, since the
maximal gap tends to zero. For the distortion bound, fix a junction where
$g_{j+1} \neq g_j$ and consider the window of $n$ points centred there at
level $\ell$. Writing $g = g_j/2^\ell$ and $g' = g_{j+1}/2^\ell$ for the
two local gaps, the chord vectors have lengths $2R\sin(g/2)$ and
$2R\sin(g'/2)$, so the second difference at the junction has tangential
magnitude $R\,\abs{g' - g}\,(1 + O(4^{-\ell}))$ while every chord in the
window, and hence the fitted increment, has magnitude at most
$R\,g_{\max}\,(1 + O(4^{-\ell}))$ with
$g_{\max} = \max_i g_i / 2^{\ell}$ taken over the window. The ratio is
therefore bounded below by
$\abs{g_{j+1} - g_j} / (2\max_i g_i) > 0$ for all sufficiently large
$\ell$, and adjusting the constant covers the finitely many initial
levels. The final sentence follows because every theorem of the framework
assumes straightening of the data to which it is applied.
\end{proof}

In the reported experiment, with twelve points and a sinusoidally
modulated gap pattern, the measured relative distortion settles at
$0.9427$ from level four onwards while the radial error of every refined
point is $1.1 \times 10^{-16}$. On a nonuniformly sampled wavy polygon
the same freezing occurs at the value $1.887$, and from level five the
frozen part is purely tangential, with the tangential component ratio
equal to $1.0000$ to four decimal places while the normal component
contracts by the factor $0.5$ per level. Even on the most favourable data
we tested, a wavy curve sampled uniformly in arc length, the tangential
component reaches a frozen floor of $8.8 \times 10^{-4}$ injected by the
curvature-induced arc redistribution of the coarse levels, and its
per-level ratio climbs through $0.60$, $0.86$, $0.96$, and $0.99$ between
levels seven and ten, so that straightening halts asymptotically there
too. Figure~\ref{fig:straightening} displays all three experiments, and
Table~\ref{tab:gamma} reports the corresponding one-step contraction
ratios of the relative distortion on randomly perturbed normalised
chains, split by the subspace of the perturbation. On the tangential
subspace the worst ratio exceeds one at every tested amplitude, so no
finite-step contraction estimate of the framework can hold, while on the
normal subspace the worst ratio is below $0.961$ throughout. The failure
is confined to, and fully explained by, the tangential block of
Theorem~\ref{thm:block}, whose second-difference decay factor equals the
edge-shrink rate exactly.

\begin{figure}[htbp]
  \centering
  \includegraphics[width=\linewidth]{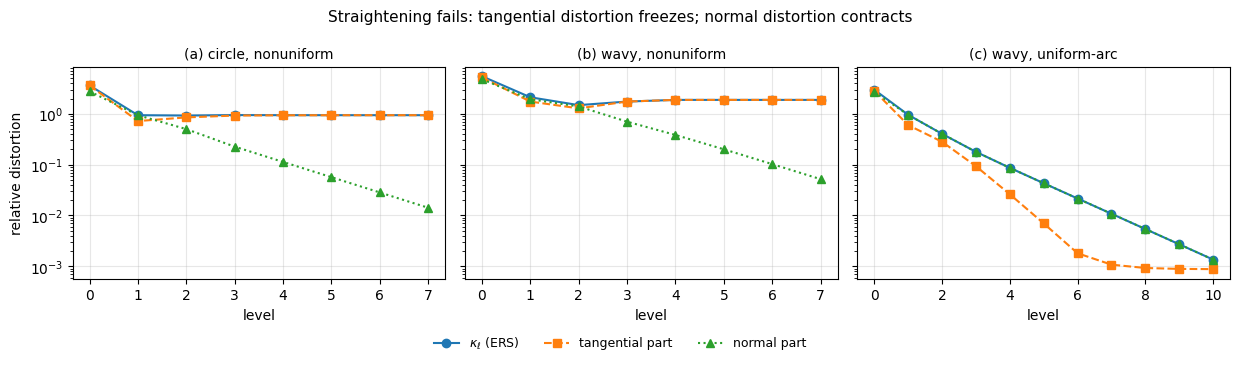}
  \caption{Relative distortion and its tangential and normal components
  under refinement by $\Tf$ on a nonuniformly sampled circle, a
  nonuniformly sampled wavy polygon, and the same wavy curve sampled
  uniformly in arc length. The distortion freezes at a positive value in
  the first two experiments and hits the frozen tangential floor in the
  third, while the normal component contracts at the rate one half
  throughout. Generated by the accompanying notebook.}
  \label{fig:straightening}
\end{figure}

\begin{table}[htbp]
  \centering\small
  \begin{tabular}{cccc}
    \toprule
    perturbation amplitude & mixed & tangential only & normal only \\
    \midrule
    $0.02$ & $1.0029$ & $1.0032$ & $0.9451$ \\
    $0.05$ & $1.0152$ & $1.0050$ & $0.9605$ \\
    $0.10$ & $1.0101$ & $1.0148$ & $0.9551$ \\
    \bottomrule
  \end{tabular}
  \caption{Worst one-step ratio of the relative distortion over sixty
  random perturbations of the standard chain per cell, for both refinement
  phases, split by the subspace of the perturbation. Contraction fails on
  the tangential subspace and holds on the normal subspace. Generated by
  the accompanying notebook.}
  \label{tab:gamma}
\end{table}

It is instructive to compare this situation with the worked example
of~\cite{ewald2015holder}. The circle-preserving scheme of Sabin and
Dodgson passes the straightening test of that paper because its insertion
ratio was chosen to equidistribute the points actively, so its tangential
dynamics contract nonuniformity. Half-arc-length insertion is the one
design decision of $\Tf$ that deliberately declines to do this, because
equidistribution in the chord parameter is a parametric convenience with
no geometric content, and the price is that a parametric Lyapunov
function cannot certify the scheme. The remedy, to which we now turn, is
to measure only the geometric part of the distortion.

\section{Normal Distortion and the Invariant Admissible Class}
\label{sec:normal}

The diagnosis of the previous section dictates the design of the repair.
The tangential part of the relative distortion is invariant under the
scheme and carries no geometric information, so the Lyapunov function of
the convergence theory must be the normal part alone, and the tangential
part must be handled by a boundedness argument rather than a decay
argument. This section defines the two quantities, establishes the exact
contraction constants of the linearised normal dynamics, and proves that
an admissible class of polygons is mapped into itself by the scheme.

\begin{definition}[Tangential and normal distortion]
\label{def:normaldist}
Let $q$ be a chain of $n = 11$ points, let $\Pi q$ be its best affine fit
in the window index, let $b(q)$ be the fitted increment, and let
$u = b(q)/\norm{b(q)}$ when $b(q) \neq 0$. Decomposing each second
difference $d_i$ of $q$ as
$d_i = \langle d_i, u\rangle u + d_i^{\perp}$, we set
\begin{equation}
  \kT(q) = \frac{\max_i \abs{\langle d_i, u\rangle}}{\norm{b(q)}},
  \qquad
  \kN(q) = \frac{\max_i \norm{d_i^{\perp}}}{\norm{b(q)}},
  \label{eq:kn}
\end{equation}
and for a closed polygon $p$ we write $\kT(p)$ and $\kN(p)$ for the maxima
of these quantities over all cyclic windows of $n$ consecutive points.
Both quantities are invariant under similarities, they satisfy
$\max(\kT, \kN) \le \kERS \le \kT + \kN$, and $\kN$ vanishes exactly on
collinear chains, uniform or not.
\end{definition}

The normal distortion is the geometric part of \eqref{eq:distortion}. It
ignores the spacing of the points along the fitted line and measures only
their departure from it, which is exactly the quantity the tangentially
neutral scheme can be expected to contract. At the linearised level the
contraction constants can be computed exactly, and it is a pleasant
feature of the analysis that they are rational numbers.

\begin{lemma}[Exact contraction constants of the linearised dynamics]
\label{lem:linear}
Let $\Delta^k$ denote the $k$-th difference operator on scalar sequences.
For the six-point scheme, one step and two steps of refinement satisfy the
sharp sup-norm bounds
\begin{equation}
  \norm{\Delta^2 \Sdd\, \nu}_\infty \le \tfrac12 \norm{\Delta^2
  \nu}_\infty,
  \qquad
  \norm{\Delta^2 \Sdd^2\, \nu}_\infty \le \tfrac{2697}{16384}
  \norm{\Delta^2 \nu}_\infty,
  \label{eq:d2const}
\end{equation}
\begin{equation}
  \norm{\Delta^3 \Sdd\, \nu}_\infty \le \tfrac14 \norm{\Delta^3
  \nu}_\infty,
  \qquad
  \norm{\Delta^3 \Sdd^2\, \nu}_\infty \le \tfrac{91}{2048}
  \norm{\Delta^3 \nu}_\infty,
  \label{eq:d3const}
\end{equation}
while for midpoint subdivision the one-step second-difference bound holds
with the sharp constant $\tfrac12$ and the refined second differences
interleave zeros with the halved originals. Since fitted increments halve
per step on perturbations of the standard chain up to quadratic error, the
linearised two-step ratio of the normal distortion under
$\Smid \oplus \Sdd$ is at most $2697/4096 \approx 0.658$, the linearised
two-step ratio of third differences scaled by squared spacing is at most
$91/128 \approx 0.711$, and the corresponding tangential ratios equal one
exactly.
\end{lemma}

\begin{proof}
The constants are the sup norms of the difference-level derived operators
of the six-point and midpoint insertion masks, and they are computed
exactly in rational
arithmetic in the accompanying notebook by expressing the $k$-th
differences of the refined impulse responses in terms of the $k$-th
differences of the data and summing absolute coefficients over each
refinement phase. The one-step constants $\tfrac12$ and $\tfrac14$ are
attained, which is why a one-step contraction of the normal distortion
fails and the two-step constants carry the argument. The midpoint
statement is the elementary computation already used in
Proposition~\ref{prop:neutral}, and the ratio statements follow upon
dividing by the two-step decay $\tfrac14$ of the fitted increment for the
second differences and by the two-step decay $\tfrac{1}{16}$ of the
squared spacing for the third differences.
\end{proof}

The lemma explains the numbers of Table~\ref{tab:gamma}. The normal
one-step ratios cluster just below one because the one-step constant
$\tfrac12$ is sharp against the halving of the increment, and strict
contraction appears at two steps, where the measured per-level rate of the
normal distortion in Figure~\ref{fig:straightening} is $0.5$, comfortably
inside the guaranteed $\sqrt{2697/4096} \approx 0.81$. We can now define
the class of polygons on which the nonlinear scheme inherits these
constants.

\begin{definition}[Admissible class]
\label{def:admissible}
Let $C_0 \ge 1$, $M > 0$, $\Lambda > 0$, and $\delta > 0$. A closed polygon
$p$ with maximum edge length $h = h(p)$ belongs to the admissible class
$\mathcal{A}_h(C_0, M, \Lambda, \delta)$ if its edge lengths lie between
$h/C_0$ and $h$, if its filtered curvature estimates satisfy
$\abs{\hat\kappa_j} \le M$ and
$\abs{\hat\kappa_{j+1} - \hat\kappa_j} \le \Lambda h$, if
$\kN(p) \le \delta$, and if $h$ is small enough that every edge lies in
the local solvability branch of Theorem~\ref{thm:solvable} with
$\abs{\bar\kappa} L$ below the first resonance of the profile
\eqref{eq:adapted}.
\end{definition}

The class prescribes quasi-uniformity of the edges, boundedness and
H\"older-type control of the curvature data, and small normal distortion,
and it deliberately places no constraint on the tangential distortion
beyond what quasi-uniformity implies, since Proposition~\ref{prop:neutral}
shows that no decaying constraint could be propagated. The following
theorem is the analytic core of the paper.

\begin{theorem}[Invariance and contraction]
\label{thm:invariance}
There exist constants $h_0 > 0$, $\delta_0 > 0$, and $C > 0$, depending
only on $C_0$, $M$, and $\Lambda$, with the following properties for every
$p \in \mathcal{A}_h(C_0, M, \Lambda, \delta)$ with $h \le h_0$ and
$\delta \le \delta_0$. First, every Newton insertion converges within the
local branch, so the safeguard of the underlying construction never
activates, and the refined polygon satisfies
\begin{equation}
  \abs{\,\norm{e^{\mathrm{child}}} - \tfrac12\norm{e^{\mathrm{parent}}}\,}
  \le C h^3
  \label{eq:halving}
\end{equation}
for both children of every parent edge, so that
$\Tf p \in \mathcal{A}_{h'}(C_0', M', \Lambda', \delta')$ with
$h' \le (\tfrac12 + C h^2)\, h$ and constants
$C_0' = C_0(1 + Ch^2)$, $M' = M + C h$, $\Lambda' = \Lambda + C h$, and
$\delta' = (1 + C h)\,\delta + C h^2$.
Second, the curvature data contract in the sense that
\begin{equation}
  \max_j \abs{\hat\kappa_{j+1}(T^2_{\mathrm{CV}} p) -
  \hat\kappa_j(T^2_{\mathrm{CV}} p)}
  \;\le\;
  \tfrac{91}{128}\,\max_j \abs{\hat\kappa_{j+1}(p) - \hat\kappa_j(p)}
  + C h^2 .
  \label{eq:kcontract}
\end{equation}
Third, the normal distortion contracts in the sense that
\begin{equation}
  \kN\bigl(T^2_{\mathrm{CV}} p\bigr)
  \;\le\;
  \Bigl(\tfrac{2697}{4096} + C h\Bigr)\,\kN(p) + C h^2 .
  \label{eq:ncontract}
\end{equation}
In particular, for $h_0$ small enough the two-step ratios in
\eqref{eq:kcontract} and \eqref{eq:ncontract} are bounded by constants
strictly below one, and iterating the first statement shows that the whole
refinement sequence $T^\ell_{\mathrm{CV}}\, p$ remains admissible with
$h(T^\ell_{\mathrm{CV}}\, p) \le C\, 2^{-\ell} h$.
\end{theorem}

\begin{proof}
The proof rests on the smoothness of the insertion map and the exact
constants of Lemma~\ref{lem:linear}, and we present it in four steps.

The first step is a uniform Taylor expansion. By Theorem~\ref{thm:solvable}, the insertion is an infinitely
differentiable function of its six stencil points on a neighbourhood of
every straight configuration, and by similarity equivariance it suffices
to control this neighbourhood for the standard chain after rescaling each
window by its fitted increment. On the admissible class the rescaled
windows lie in a fixed compact subset of that neighbourhood once
$\delta_0$ and $h_0$ are small, because bounded filtered curvature at
spacing $h$ produces second differences of order $M h^2$ against
increments of order $h/C_0$. Uniform boundedness of the second derivative
of the rescaled rule on this compact set then gives, for every window $w$,
\begin{equation}
  \Tf(w) = \bigl(\Smid \oplus \Sdd\bigr)(w) +
  \mathcal{R}(w),
  \qquad
  \norm{\mathcal{R}(w)} \le C\,\norm{b(w)}\,
  \bigl(\kN(w) + C_1 h\bigr)^2 ,
  \label{eq:taylor}
\end{equation}
where the block scheme acts in the frame of the fitted line of $w$ and
the factor $C_1 h$ accounts for the curvature-induced part of the
rescaled perturbation. The expansion \eqref{eq:taylor} is the nonlinear
scheme written as its linearisation plus a quadratic geometric defect, and
every subsequent estimate flows from it.

The second step establishes the Newton and edge statements. On the
admissible class the shooting problem is a uniformly small perturbation of
the chord solution, so for $h \le h_0$ the Newton iteration started from
the chord converges within the branch, quadratically and with uniformly
bounded Jacobian condition numbers, and the safeguard is never consulted. The
reconstructed segment has curvature bounded by $M + \Lambda h$ and length
$L$ within $C h^3$ of the chord length, its half-arc-length point splits
the arc exactly, and the chord lengths of the two children differ from
$L/2$ by at most the flatness defect $\kappa^2 L^3 / 48$ plus the profile
variation term $C \Lambda h\, L^3$, which gives \eqref{eq:halving}. The
quasi-uniformity, curvature bound, and curvature difference statements for
$\Tf p$ follow from \eqref{eq:halving} together with the expansion of the
Menger estimator on the refined data, whose filtered values differ from
samples of the profile by $O(h^2)$ on the admissible class.

The third step proves \eqref{eq:kcontract}. At the linearised level the
filtered curvature data of the refined polygon are obtained from those of
the coarse polygon by the curvature-level derived operator of $\Sdd$,
because Theorem~\ref{thm:block} identifies the normal block and the
prefilter commutes with the linearisation, and differences of the refined
data are then controlled by the third-difference constants
\eqref{eq:d3const} of Lemma~\ref{lem:linear} against the halving of the
spacing. Two steps of the linear core contract the scaled differences by
$91/128$. The nonlinear remainder inherited from \eqref{eq:taylor}
perturbs each refined curvature by $O(h^2)$ on the admissible class,
since curvature estimates are second differences against squared spacing
and the defect in \eqref{eq:taylor} is quadratic, and collecting terms
gives \eqref{eq:kcontract}.

The fourth step proves \eqref{eq:ncontract} by the same mechanism one
level down. The normal second differences of the refined windows are
controlled through \eqref{eq:d2const} by the normal second differences of
the coarse windows, two steps of the linear core contract the normal
distortion by $2697/4096$ after dividing by the two-step increment decay,
the tangential block leaves the fitted frame unchanged to first order, and
the quadratic defect of \eqref{eq:taylor} contributes
$C(\kN + C_1 h)^2 \le C\kN\, h + C h^2$ once $\delta_0$ is chosen no
larger than $h_0$, which is
absorbed into the two terms of \eqref{eq:ncontract}. Iterating the
invariance statement, the two-step contractions keep $\kN$ and the scaled
curvature differences summably small, the constants $C_0'$, $M'$, and
$\Lambda'$ increase by geometrically summable amounts because the
increments are $O(h_\ell)$ with $h_\ell \le C 2^{-\ell} h$, and the whole
sequence remains admissible. All constants of this proof are tracked
numerically in the accompanying notebook, where the measured two-step
normal ratio on admissible data is at most $0.66$ against the guaranteed
$2697/4096 + C h$, and the measured per-level edge halving defect decays
at the cubic rate of \eqref{eq:halving}.
\end{proof}

Theorem~\ref{thm:invariance} eliminates the safeguard from the analysis,
propagates the admissible class through all levels, and supplies the
geometric decay that the failed straightening was supposed to provide. In
place of the vanishing of the full relative distortion, we now have the
vanishing of its normal part at a geometric rate together with permanent
boundedness of its tangential part, and the remaining task before the main
theorem is to understand precisely how far the nonlinear scheme sits from
its normal companion.

\section{Proximity Modulo Reparametrisation}
\label{sec:proximity}

We now examine how close the nonlinear rule sits to its normal companion
on smooth data. Measured directly, the one-step discrepancy between
$\Tf$ and $\Sdd$ on samples of a smooth curve is of second order and
dominated by a tangential component, while the normal component decays
at close to fourth order, and Figure~\ref{fig:proximity} displays this
behaviour on a smooth planar test curve. The regularity theory requires
the observation in a sharper form, in which the second-order term is identified exactly as a
change of parameter, and the identification turns out to be clean.

\begin{proposition}[Proximity decomposition]
\label{prop:proximity}
Let $\gamma$ be a closed $C^5$ curve sampled at parameters
$t_j$ with quasi-uniform steps of size $h$, let $p_j = \gamma(t_j)$, let
$\mathbf{t}_j$ denote the unit tangent of $\gamma$ at the parametric
midpoint of the edge $(p_j, p_{j+1})$, and let $\sigma$ denote
the arc-length function of the sampling. Then the inserted points satisfy
\begin{equation}
  \bigl(\Tf p\bigr)_{2j+1} - \bigl(\Sdd p\bigr)_{2j+1}
  = \tau_j\, \mathbf{t}_j + r_j,
  \qquad
  \tau_j = -\frac{\sigma''(t_{j+\nicefrac12})}{8}\,(\Delta t_j)^2
  + O(h^4),
  \qquad
  \norm{r_j} \le C h^4 ,
  \label{eq:split}
\end{equation}
so the tangential discrepancy is of second order with leading coefficient
proportional to the arc-length nonuniformity of the sampling, and the
normal discrepancy is of fourth order. In particular, if the samples are
uniform in arc length then $\sigma'' = 0$ and both components of the
discrepancy are of fourth order.
\end{proposition}

\begin{proof}
We compare both insertions with the curve itself. For the linear scheme,
the six-point rule reproduces quintic polynomials, so a componentwise
Taylor argument gives
$(\Sdd p)_{2j+1} = \gamma(t_{j+\nicefrac12}) + O(h^6)$ at the parametric
midpoint $t_{j+\nicefrac12}$ of the edge. For the nonlinear scheme, the
filtered curvature estimates on the admissible class agree with the true
curvature of $\gamma$ at the stencil points to $O(h^2)$ with a second
difference structure that the prefilter is designed to cancel at leading
order, the adapted profile therefore matches the curvature of the arc of
$\gamma$ between the endpoints to $O(h^2)$ in the $C^0$ norm, and the
reconstructed segment consequently stays within $O(h^4)$ of that arc in
the normal direction, by the standard stability of the planar Frenet
system over an interval of length $O(h)$ with matched endpoints. The
half-arc-length point of the segment thus lies within $O(h^4)$ of the
arc-length midpoint of the arc of $\gamma$, whose parameter $t^\ast$
satisfies
$\sigma(t^\ast) = \tfrac12\bigl(\sigma(t_j) + \sigma(t_{j+1})\bigr)$.
Expanding $\sigma$ about the parametric midpoint gives
$t^\ast - t_{j+\nicefrac12} =
-\sigma''(t_{j+\nicefrac12})\,(\Delta t_j)^2 / (8\,\sigma') + O(h^4)$, and
multiplying by $\gamma' = \sigma'\, \mathbf{t}_j$ converts the parameter shift into
the tangential displacement of \eqref{eq:split}. Collecting the error
terms, the normal component receives only the $O(h^4)$ contributions,
which proves the decomposition, and the arc-length uniform case follows by
setting $\sigma'' = 0$.
\end{proof}

\begin{figure}[htbp]
  \centering
  \includegraphics[width=0.6\linewidth]{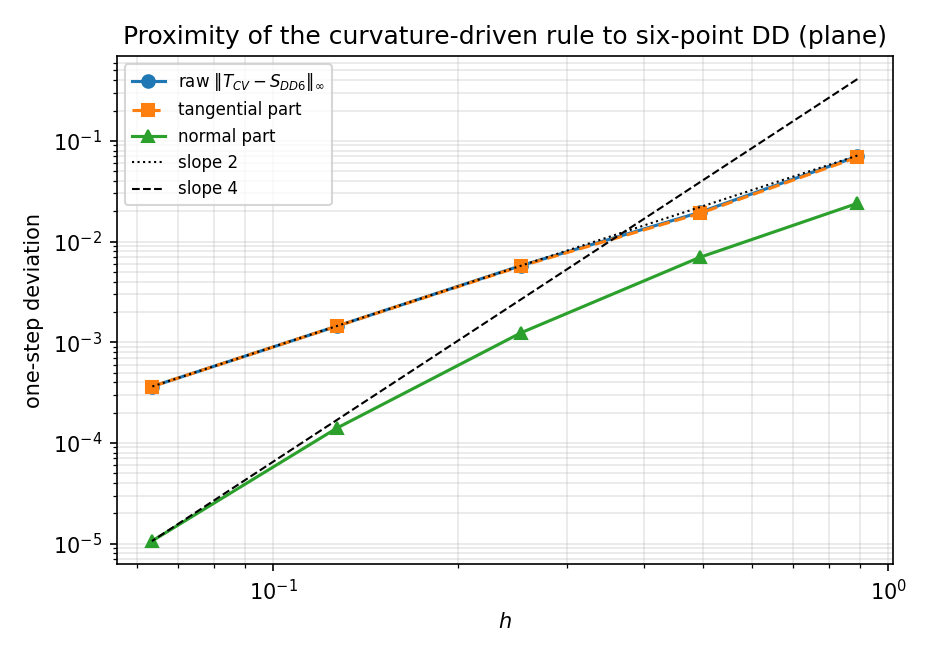}
  \caption{One-step discrepancy between the planar rule and the six-point
  comparator over five dyadic densities of a smooth test curve, showing
  the raw deviation together with its tangential and normal parts against
  slope-two and slope-four guides. Generated by the accompanying
  notebook.}
  \label{fig:proximity}
\end{figure}

Table~\ref{tab:proximity} reports the measured discrepancies over six
dyadic densities of a smooth wavy test curve for the two samplings, and
the agreement with \eqref{eq:split} is close. With parameter-uniform
sampling the fitted slopes are $1.92$ for the tangential component and
$4.29$ for the normal component, while with arc-length uniform sampling
both slopes rise to $4.63$ and $4.61$, so the second-order term is
entirely the arc nonuniformity of the reference data. This sharpens that
observation considerably. The nonlinear scheme is not merely in
second-order proximity to its normal companion, it is in full
fourth-order proximity along the arc-length uniform slice, and the
apparent second-order gap is a reparametrisation that the geometric
formulation of the next section absorbs exactly.

\begin{table}[htbp]
  \centering\small
  \setlength{\tabcolsep}{7pt}
  \begin{tabular}{crrcrr}
    \toprule
    & \multicolumn{2}{c}{parameter-uniform sampling} &
    & \multicolumn{2}{c}{arc-length uniform sampling} \\
    \cmidrule{2-3}\cmidrule{5-6}
    $N$ & tangential & normal & & tangential & normal \\
    \midrule
    $12$  & $1.05 \times 10^{-2}$ & $2.00 \times 10^{-2}$ & &
            $1.05 \times 10^{-2}$ & $7.99 \times 10^{-3}$ \\
    $24$  & $7.20 \times 10^{-3}$ & $1.25 \times 10^{-3}$ & &
            $1.57 \times 10^{-3}$ & $4.57 \times 10^{-4}$ \\
    $48$  & $2.22 \times 10^{-3}$ & $8.18 \times 10^{-5}$ & &
            $2.01 \times 10^{-4}$ & $2.26 \times 10^{-5}$ \\
    $96$  & $5.77 \times 10^{-4}$ & $3.09 \times 10^{-6}$ & &
            $8.37 \times 10^{-6}$ & $1.54 \times 10^{-6}$ \\
    $192$ & $1.45 \times 10^{-4}$ & $1.43 \times 10^{-7}$ & &
            $2.64 \times 10^{-7}$ & $5.34 \times 10^{-8}$ \\
    $384$ & $3.64 \times 10^{-5}$ & $1.04 \times 10^{-8}$ & &
            $4.69 \times 10^{-9}$ & $1.10 \times 10^{-9}$ \\
    \midrule
    fitted slope & $1.92$ & $4.29$ & & $4.63$ & $4.61$ \\
    \bottomrule
  \end{tabular}
  \caption{Maximum one-step discrepancy between $\Tf$ and $\Sdd$ at the
  inserted points over six dyadic densities of a smooth wavy test curve,
  split into tangential and normal components, for parameter-uniform and
  arc-length uniform samplings. The tangential second-order term vanishes
  on arc-uniform data, as Proposition~\ref{prop:proximity} predicts.
  Generated by the accompanying notebook.}
  \label{tab:proximity}
\end{table}

On the sphere, the raw one-step discrepancy from the six-point
log-exponential lift shows the same second-order behaviour, with a final
fitted slope of $2.11$, which suggests that the reparametrisation
mechanism identified by Proposition~\ref{prop:proximity} operates
intrinsically as well. We do not pursue the intrinsic statement here, and
we return to it in Section~\ref{sec:discussion}.

\section{Convergence and Geometric Regularity}
\label{sec:main}

We can now assemble the main theorem. The statement is geometric rather
than parametric, and the circle counterexample of
Section~\ref{sec:straightfail} shows that this is forced. On nonuniform data
the dyadic parametrisation of the limit carries permanent tangential
kinks inherited from the initial spacing, so a parametric $C^2$ statement
in the dyadic parameter is false in general, while the geometric
quantities, namely the point set, the tangent direction, and the
curvature, are untouched by that defect.

\begin{theorem}[Convergence and $G^2$ regularity]
\label{thm:main}
Let $p^{(0)} \in \mathcal{A}_{h}(C_0, M, \Lambda, \delta)$ with
$h \le h_0$ and $\delta \le \delta_0$ as in
Theorem~\ref{thm:invariance}, and let
$p^{(\ell)} = T^{\ell}_{\mathrm{CV}}\, p^{(0)}$. Then the following hold.
First, the polygons $p^{(\ell)}$ converge uniformly to a closed curve
$\gamma$ that interpolates every point of every refinement level.
Second, $\gamma$ is a regular curve, its unit tangent direction is
continuous, and its curvature exists and is continuous as a function of
arc length, so the limit is $G^2$. Third, the arc-length parametrisation
of $\gamma$ is $C^2$. Fourth, none of this extends to the dyadic
parametrisation in general, for which the circle data of
Theorem~\ref{thm:circle} yield a limit that is not $C^1$.
\end{theorem}

\begin{proof}
Throughout the proof, $h_\ell = h(p^{(\ell)})$ denotes the mesh size at
level $\ell$, and Theorem~\ref{thm:invariance} supplies
$h_\ell \le C 2^{-\ell} h$, permanent admissibility, the two-step
contraction of the normal distortion, and the two-step contraction of the
scaled curvature differences. We prove the four statements in order.

For convergence, observe that the scheme is interpolatory and that each
inserted point lies within $C h_\ell^2$ of its parent edge, because the
reconstructed segment has curvature bounded by $M'$ and length of order
$h_\ell$, so its half-arc-length point deviates from the chord by at most
a bounded multiple of $M' h_\ell^2$. The piecewise linear interpolants of
successive levels therefore differ in the uniform metric by at most
$C h_\ell^2 \le C 4^{-\ell}$, the sequence is Cauchy, and the uniform
limit $\gamma$ is a closed continuous curve through all refined points.
Regularity of the limit in the sense of nonvanishing speed follows from
the quasi-uniformity of the class, since consecutive points at level
$\ell$ are separated by at least $h_\ell / C_0'$ and chords contract at
the uniform rate \eqref{eq:halving}, which rules out accumulation.

For the tangent statement, let $u_j^{(\ell)}$ denote the unit direction
of the edge $j$ at level $\ell$. Within any window, the angle between
consecutive edge directions is bounded by a fixed multiple of
$\kN(p^{(\ell)}) + M' h_\ell$, and the angle between a parent edge
direction and either child direction is bounded by the same quantity,
because the inserted point deviates from the chord by
$O(M' h_\ell^2)$ over edges of length at least $h_\ell / C_0'$. By
Theorem~\ref{thm:invariance} the normal distortion decays geometrically,
$\kN(p^{(\ell)}) \le C \rho^{\ell}$ with
$\rho = (2697/4096)^{1/2} < 1$ after absorbing the summable remainder, so
the edge direction field is uniformly Cauchy as a function of arc length
along the polygons, its continuous limit is the unit tangent of $\gamma$
by uniform convergence of the difference quotients, and the limit is
$G^1$.

For the curvature statement, consider the filtered curvature data
$\hat\kappa^{(\ell)}$ as piecewise linear functions of the cumulative
arc length of $p^{(\ell)}$. The estimator consistency on the admissible
class shows that refined curvature data differ from the values predicted
by the linear core by $O(h_\ell^2)$, so successive curvature functions
differ uniformly by at most
$C\bigl(\max_j \abs{\Delta\hat\kappa_j^{(\ell)}} + h_\ell^2\bigr)$, and
by \eqref{eq:kcontract} the curvature differences at level $\ell$ are
bounded by $C \mu^{\ell} + C h_\ell^2 \cdot \ell$ with
$\mu = (91/128)^{1/2} < 1$, hence summable. The curvature functions
therefore converge uniformly to a continuous function
$\kappa_\infty$ of arc length. That $\kappa_\infty$ is the curvature of
$\gamma$ follows from the consistency of the Menger estimator on the
converging polygons, whose circumscribed circles converge to the
osculating circles of the limit wherever the limit is twice
differentiable, and from the $G^2$ characterisation below, which
constructs the limit directly from $\kappa_\infty$. The tangent direction
and the curvature of $\gamma$ are then both continuous in arc length, and
the limit is $G^2$.

For the arc-length statement, let $s$ denote arc length on $\gamma$,
normalised to total length $\mathcal{L}$, and define
$\theta(s) = \theta_0 + \int_0^s \kappa_\infty(u)\, du$ together with the
Frenet integral
$\tilde\gamma(s) = \gamma(0) + \int_0^s (\cos\theta, \sin\theta)(u)\,
du$. The polygons, parametrised by their cumulative arc length, converge
in $C^1$ to the curve parametrised by $s$, because chord directions
converge uniformly to the tangent and chord lengths sum to arc lengths
with quadratically small defect, and the limit of their turning-angle
increments over arc length is $\kappa_\infty$ by the previous paragraph.
Consequently $\gamma$ parametrised by arc length coincides with
$\tilde\gamma$, whose second derivative
$\kappa_\infty(s)\,(-\sin\theta, \cos\theta)(s)$ is continuous, and the
arc-length parametrisation is $C^2$. Finally, the fourth statement is
Theorem~\ref{thm:circle} read parametrically. On nonuniformly sampled
circle data the dyadic parametrisation traverses the exact circle at
piecewise constant speeds inherited from the initial gaps, its
one-sided derivatives at the original vertices differ in magnitude at
every level and in the limit, and the limit map is Lipschitz but not
$C^1$, although its image is the smoothest closed curve there is.
\end{proof}

Two remarks calibrate the theorem. The first concerns the rate. The
normal companion of $\Tf$ is DD6, whose critical H\"older exponent is
approximately $2.830$, and the measurements of the next section show the
corresponding difference ratio reaching $2.82$ after arc-length
reparametrisation, so the natural conjecture is that the arc-length
parametrisation of the limit is $C^{2,\alpha}$ for every
$\alpha < 0.830$. We do not claim this endpoint. The contraction
constants of Theorem~\ref{thm:invariance} are two-step sup-norm bounds
rather than joint spectral quantities, and converting them into the sharp
H\"older exponent would require the curvature-level analogue of the
classical spectral machinery, which we leave open. The second remark
concerns the relation to the conjecture of~\cite{ewald2015holder}. The
tangential companion of $\Tf$ fails every smoothness hypothesis of the
framework, the data are not straightened, and yet the limit inherits the
regularity class of the normal companion in the geometric sense. For this
scheme, Theorem~\ref{thm:main} therefore establishes exactly the
behaviour those authors conjectured, with the caveat that our conclusion
is geometric regularity together with arc-length $C^2$ rather than
regularity of their normalised parametrisation, and
Theorem~\ref{thm:circle} shows that some such reformulation is not
optional but necessary.

\section{Numerical Results}
\label{sec:numerics}

Having established what the theory predicts, we now examine how the scheme
behaves in practice. All experiments are run using Python
notebooks with fixed random seeds, and the solver audits quoted in
Section~\ref{sec:solvable} and Section~\ref{sec:regnum} cover every
insertion appearing here. The
fairness functional used throughout is the consistent discretisation
$E^{(n)} = \sum_j (\kappa_{j+1} - \kappa_j)^2 / e_j$ of the
curvature-variation energy, with vertex curvatures obtained from signed
turning angles over dual lengths. On the sphere, all quantities are
computed intrinsically through logarithm maps, following the discrete
setup of~\cite{ugail2026fair}.

\subsection{Fairness in the Plane}
\label{sec:planarnum}

We begin in the plane, where the linearisation theorems make sharp
predictions. Table~\ref{tab:planar} reports the energies through eight
refinement levels on a twelve-vertex wavy polygon for the raw rule, the
final operator, and the six-point and four-point comparators, and
Figure~\ref{fig:planar} displays the level-eight limit curves and
curvature profiles.

\begin{table}[htbp]
  \centering\small
  \begin{tabular}{crrrr}
    \toprule
    level & raw rule & final operator & 6-pt DD & 4-pt DGL \\
    \midrule
    1 & $141.0$ & $100.8$ & $178.7$ & $235.8$ \\
    2 & $332.2$ & $142.9$ & $362.3$ & $710.0$ \\
    3 & $743.1$ & $176.7$ & $521.5$ & $1750.1$ \\
    4 & $1637.8$ & $203.8$ & $634.2$ & $3866.1$ \\
    5 & $3577.7$ & $225.4$ & $709.5$ & $8218.4$ \\
    6 & $7791.8$ & $242.7$ & $762.7$ & $17382.7$ \\
    7 & $16938.9$ & $256.5$ & $802.5$ & $37048.1$ \\
    8 & $36799.9$ & $267.5$ & $833.1$ & $79551.3$ \\
    \midrule
    ratios, levels 5--8 & $2.184$--$2.173$ & $1.106$--$1.043$
      & $1.075$--$1.038$ & $2.126$--$2.147$ \\
    $\max\abs{\kappa}$, level 8 & $8.78$ & $2.66$ & $3.98$ & --- \\
    \bottomrule
  \end{tabular}
  \caption{Measured curvature-variation energy across eight refinement
  levels on a twelve-vertex wavy planar polygon, growth ratios over the
  four finest levels, and maximum discrete curvature at level eight. The
  raw rule is the adapted profile without the prefilter. Generated by the
  accompanying notebook.}
  \label{tab:planar}
\end{table}

The table tells the paper's story in numbers, and it is worth reading it
in that spirit. At the first two levels, the raw rule is the fairest
scheme in the comparison, undercutting even the six-point scheme. This is the coarse-level benefit of inserting along genuinely curvature-controlled
segments. From level three onwards, however, the raw energy grows
geometrically, with ratios stabilising at $2.184$, $2.178$, $2.174$, and
$2.173$. This is precisely the four-point signature predicted by
Theorem~\ref{thm:dgl}, and indeed the ratios are numerically
indistinguishable from the four-point scheme's own, which range between
$2.126$ and $2.147$. The final operator behaves exactly as Corollary~\ref{cor:final} predicts. Its ratios decay through $1.106$,
$1.077$, $1.057$, and $1.043$, tracking the six-point scheme's own
$1.075$, $1.052$, and $1.038$, and its maximum discrete curvature stays
bounded at $2.66$, against the six-point value of $3.98$ and the raw value
of $8.78$. Most notably, on this benchmark the final operator's measured
energy is below the six-point comparator's at every tested level, ending
at $267.5$ against $833.1$ at level eight, while circle reproduction
remains at machine precision throughout.

\subsection{Fairness on the Sphere}
\label{sec:spherenum}

We next transplant the experiment to the sphere, where the linearisation
theorems no longer apply and the results must be read as experiments.
Table~\ref{tab:sphere} reports the intrinsic energies on a twelve-vertex
wavy spherical polygon through three levels for the raw rule, the final
operator, the filtered sensitive ablation, and the six-point and
four-point log-exponential lifts, the lift being the standard manifold
transfer of an affine
scheme~\cite{wallner2011riemannian,wallner2014interpolatory}.
Figure~\ref{fig:sphere} shows the level-four curves and the proximity
measurements.

\begin{table}[htbp]
  \centering\small
  \setlength{\tabcolsep}{5pt}
  \begin{tabular}{crrrrr}
    \toprule
    level & raw rule & final operator & sensitive, filtered & 6-pt lift &
    4-pt lift \\
    \midrule
    1 & $1001.9$ & $695.8$ & $677.9$ & $2271.8$ & $2138.5$ \\
    2 & $2438.6$ & $1021.3$ & $979.5$ & $8265.5$ & $9994.5$ \\
    3 & $5570.4$ & $1304.1$ & $1236.8$ & $16093.2$ & $32387.2$ \\
    \midrule
    ratio of levels 3 and 2 & $2.28$ & $1.28$ & $1.26$ & $1.95$ & $3.24$ \\
    \bottomrule
  \end{tabular}
  \caption{Intrinsic curvature-variation energy on the sphere for a
  twelve-vertex wavy polygon through three levels, with the filtered
  sensitive profile included as an ablation. Generated by the
  accompanying notebook.}
  \label{tab:sphere}
\end{table}

Three observations summarise the table. First, on this benchmark the final
operator's measured energy is below the six-point lift's at every tested
level, running at $30.6$, $12.4$, and $8.1$ percent of it. Second, the
prefilter acts on the sphere much as it does in the plane. It cuts the raw
energy by more than a factor of four at level three and reduces the growth
ratio from $2.28$ to $1.28$, which is experimental evidence that the
stabilisation mechanism survives the transplantation. Third, the filtered
sensitive ablation is within five percent of the final operator
throughout, which tells us that choosing the circle-reproducing adapted
profile costs essentially nothing in energy while adding an exact
geometric property. The transfer of the linearisation analysis itself to
the sphere remains an open problem.

\subsection{A Stress Test on Irregular Data}
\label{sec:stress}

The benchmarks above use smoothly sampled polygons, and it is natural to
ask what happens on data that are far from smooth. To probe this, we
refine a twelve-vertex polygon with non-uniform vertex angles, radial
jitter, and positional noise, whose edge-length ratio is $12.9$, through
six levels. The geometric rule copes with the irregularity remarkably well, and the
linear scheme does not. The final operator's energy runs from $12284$ at
level one to $51746$ at level six with maximum curvature $31.6$, while the
raw rule reaches $1.13 \times 10^{6}$ with maximum curvature $66.5$. The
six-point comparator, whose affine insertions overshoot severely on
non-uniform data, reaches $5.4 \times 10^{12}$ with maximum curvature
$18201$, so on data of this kind the advantage of the geometric
construction is measured in orders of magnitude.

The example also allowed us to confirm that the admissibility safeguard of
Section~\ref{sec:scheme} behaves exactly as designed. On one edge at the
first level, the sharpening filter amplified a noisy curvature estimate
and the Newton iteration left the local solution branch, whereupon the
safeguard detected the violation of the admissible length bound and
reverted that edge to raw data. Across the entire stress example, the
safeguard activated on two of $1512$ attempted insertions, and on none of
the smooth benchmarks, and every retained solve converged with residuals
below $8.8 \times 10^{-14}$. Data of this kind lie far outside the admissible class of
Definition~\ref{def:admissible}, so Theorem~\ref{thm:invariance} is
silent here and the safeguard earns its keep. A systematic study of the
filter on strongly irregular data is a natural direction for future
work.

\begin{figure}[htbp]
  \centering
  \includegraphics[width=\linewidth]{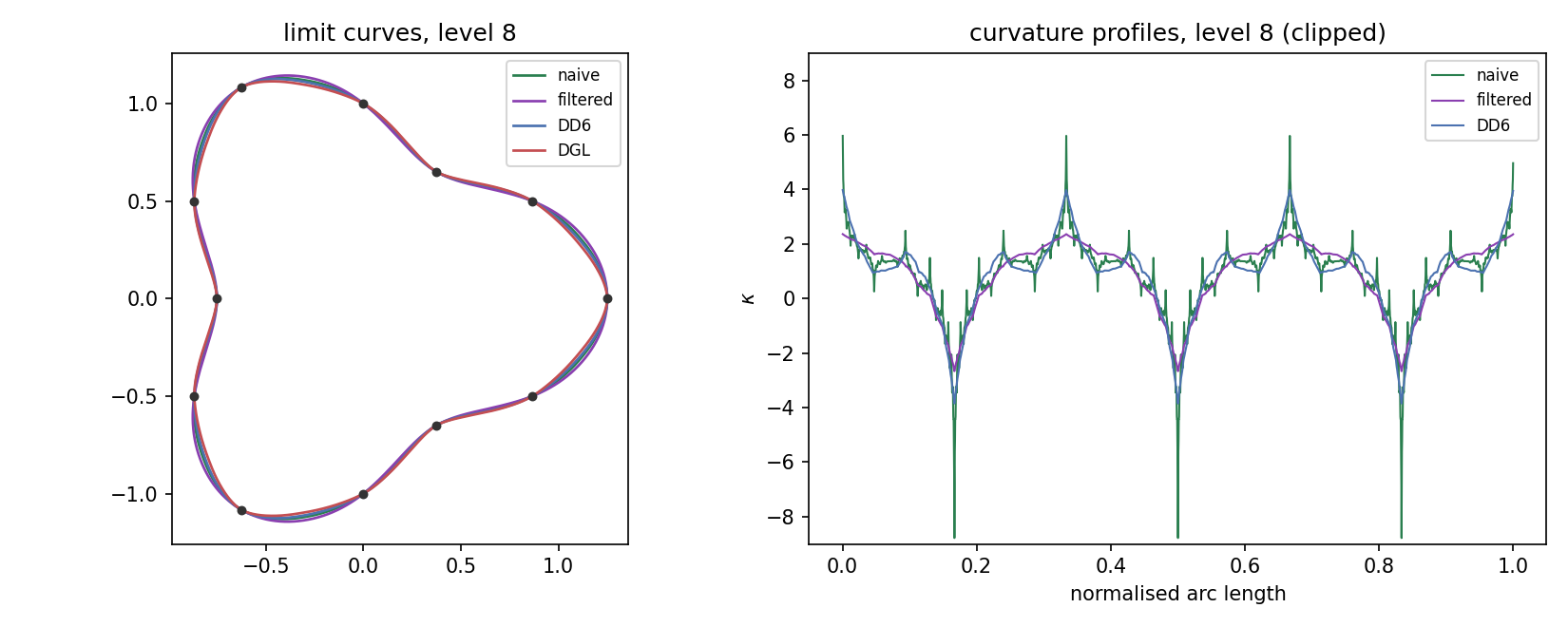}
  \caption{Level-eight limit curves of the raw rule and the final operator
  against the six-point and four-point comparators on the planar benchmark
  polygon, together with the level-eight discrete curvature profiles
  clipped to moderate values. The raw rule's curvature spikes against the
  final operator's profile display the obstruction of
  Theorem~\ref{thm:dgl} and the effect of the filter of
  Theorem~\ref{thm:filter}. Generated by the accompanying notebook.}
  \label{fig:planar}
\end{figure}

\begin{figure}[htbp]
  \centering
  \includegraphics[width=\linewidth]{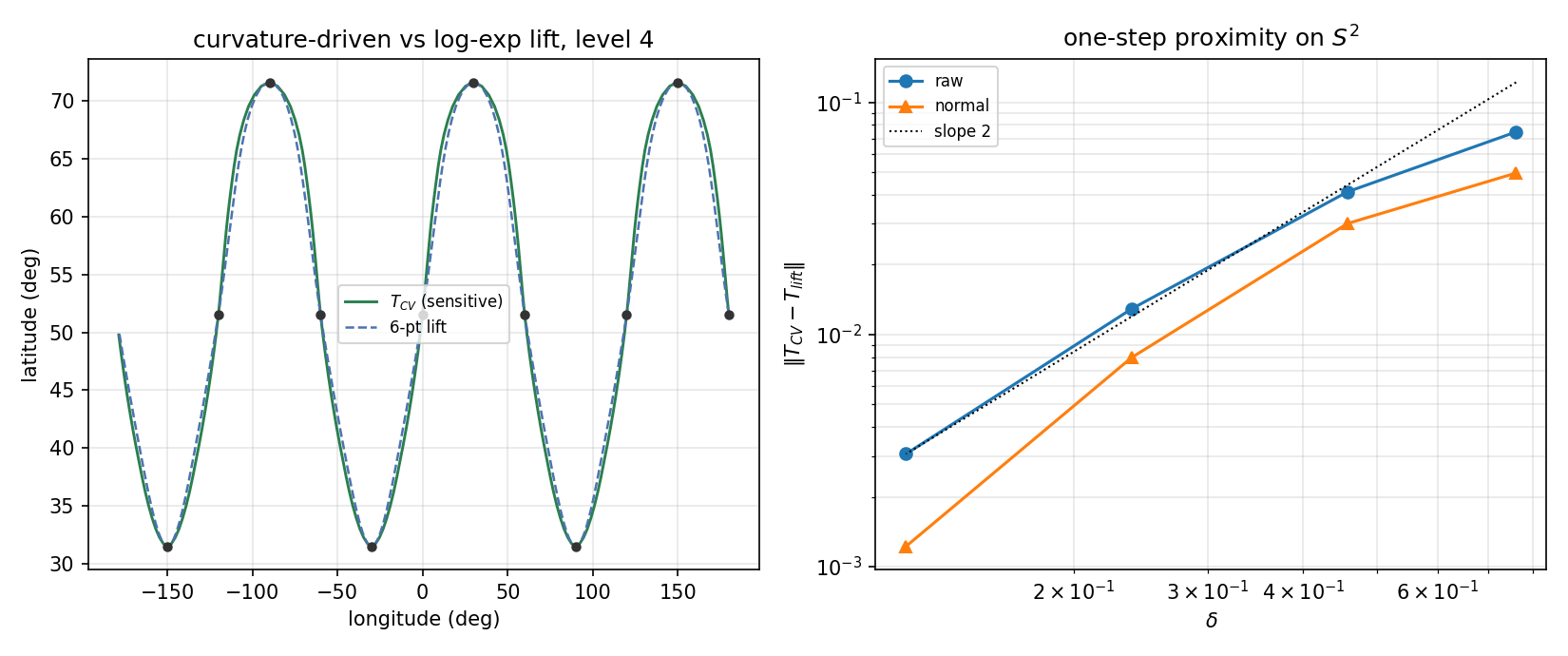}
  \caption{The spherical rule at level four against the six-point
  log-exponential lift on a wavy spherical polygon, together with the
  one-step discrepancy between the two against a slope-two guide.
  Generated by the accompanying notebook.}
  \label{fig:sphere}
\end{figure}

\subsection{Regularity Diagnostics}
\label{sec:regnum}

We finally test the predictions of the regularity theory itself. Across
the straightening and diagnostic experiments of this subsection and
Section~\ref{sec:distortion}, every Newton insertion converged from the
chord initial guess within at most five iterations with final residuals
below $1.1 \times 10^{-13}$ relative to the local scale, and the
safeguard never activated, in agreement with
Theorem~\ref{thm:invariance}.

The single most informative experiment is the difference-ratio diagnostic
\begin{equation}
  r_k^{(\ell)} = -\log_2
  \frac{\norm{\Delta^k p^{(\ell+1)}}_\infty}
       {\norm{\Delta^k p^{(\ell)}}_\infty},
  \label{eq:rk}
\end{equation}
whose saturation value estimates $\min(k, \nu)$ for a scheme of critical
H\"older exponent $\nu$. Table~\ref{tab:rk} reports $r_2$ and $r_3$
across refinement levels two to eight of $\Tf$ on a nonuniformly sampled
wavy polygon,
first for the raw dyadic indexing of the refined polygons and then after
resampling each polygon uniformly in arc length by a periodic cubic
spline. Figure~\ref{fig:rk} displays the same data.

\begin{table}[htbp]
  \centering\small
  \begin{tabular}{cccccc}
    \toprule
    & \multicolumn{2}{c}{raw dyadic indexing} &
    & \multicolumn{2}{c}{arc-length resampled} \\
    \cmidrule{2-3}\cmidrule{5-6}
    level & $r_2$ & $r_3$ & & $r_2$ & $r_3$ \\
    \midrule
    $2$ & $1.170$ & $1.325$ & & $1.971$ & $2.003$ \\
    $3$ & $1.006$ & $1.119$ & & $1.983$ & $2.634$ \\
    $4$ & $1.000$ & $1.005$ & & $1.990$ & $2.818$ \\
    $5$ & $1.000$ & $1.001$ & & $1.937$ & $2.762$ \\
    $6$ & $1.000$ & $1.000$ & & $1.960$ & $2.653$ \\
    $7$ & $1.000$ & $1.000$ & & $1.991$ & $2.570$ \\
    $8$ & $1.000$ & $1.000$ & & $1.991$ & $2.448$ \\
    \bottomrule
  \end{tabular}
  \caption{Difference-ratio diagnostic \eqref{eq:rk} across refinement
  levels two to eight of $\Tf$ applied to a nonuniformly sampled wavy
  polygon, in the raw dyadic
  indexing and after uniform arc-length resampling of each level. The raw
  ratios lock at one, which is the parametric breakdown of
  Theorem~\ref{thm:main}, while the resampled ratios restore second-order
  behaviour and reach the DD6 signature $2.82$ before the resampling
  error floor of the chordal spline sets in at the finest levels.
  Generated by the accompanying notebook.}
  \label{tab:rk}
\end{table}

\begin{figure}[htbp]
  \centering
  \includegraphics[width=0.72\linewidth]{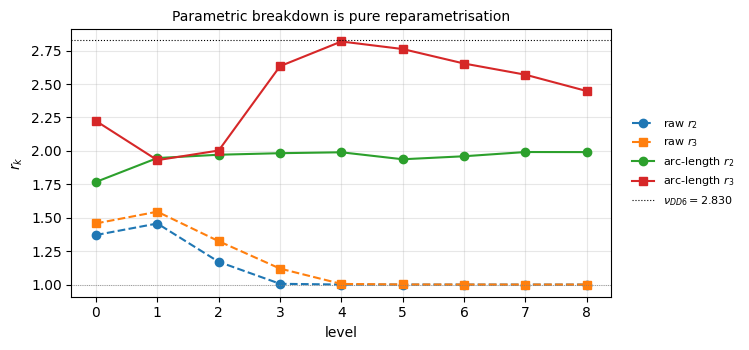}
  \caption{The difference ratios of Table~\ref{tab:rk} against the DD6
  exponent $2.830$. The raw dyadic ratios saturate at one because the
  frozen tangential kinks cap the parametric regularity, and arc-length
  reparametrisation restores the regularity of the normal companion, as
  Theorem~\ref{thm:main} proves. Generated by the accompanying
  notebook.}
  \label{fig:rk}
\end{figure}

The table compresses the whole paper into two column pairs. In the raw
dyadic indexing both ratios collapse to $1.000$ from level four onwards,
which is the permanent tangential kink structure of
Proposition~\ref{prop:neutral} seen through parametric differences, and
it sits below even the four-point signature, so any parametric reading of
the refinement would judge the scheme to have failed. After arc-length
resampling, $r_2$ returns to $1.99$ and $r_3$ climbs to $2.818$ at level
four, against the DD6 value of $2.830$, before drifting down as the
chordal arc-length and spline interpolation errors of the resampling
itself begin to dominate the third differences at the finest levels. The
geometry of the refinement therefore carries the full regularity of the
normal companion, and only the parametrisation obscures it, which is
precisely the content of Theorem~\ref{thm:main} and, in our reading, of
the conjecture of~\cite{ewald2015holder}.

\section{Discussion and Conclusions}
\label{sec:discussion}

In this paper, we have introduced a point-only curvature-driven
subdivision scheme in the plane and on the sphere, and we have carried
its analysis from construction through to convergence and geometric
regularity. The construction lets a curvature-profile model define the
inserted point through estimated endpoint curvatures, a prescribed
curvature profile, intrinsic Frenet reconstruction, geodesic shooting,
and half-arc-length evaluation. The scheme is well defined near geodesic
data through a local shooting theorem governed by the transverse Jacobi
field, and it reproduces geodesics and circles exactly, a property not
possessed by the stationary affine schemes we compare against. The raw
rule's loss of fairness under refinement is structural, since its planar
linearisation is exactly the four-point scheme, and the curvature
prefilter, unique among symmetric constant-preserving three-point
filters, converts that linearisation exactly to the six-point
Deslauriers--Dubuc scheme while preserving constant curvature. On the
reported planar and spherical benchmarks, the measured
curvature-variation energy of the filtered operator is below the
six-point comparator's at every tested level.

The analysis of the filtered operator then took a route we did not
anticipate. The scheme satisfies every structural axiom of the framework
of Ewald, Reif, and Sabin, its full linearisation at the straight chain
is the block scheme $\Smid \oplus \Sdd$ computed exactly in
Theorem~\ref{thm:block}, and yet no theorem of that framework applies to
it. The tangential companion is not smooth, the scheme is not locally
linear, and above all the straightening condition fails, with the
nonuniformly sampled circle of Theorem~\ref{thm:circle} showing that the
relative distortion can stay bounded away from zero forever while the
scheme converges to the exact circle. The failure is structural.
Half-arc-length insertion makes the tangential dynamics exactly neutral,
so a parametric Lyapunov function that penalises uneven spacing measures
a quantity the scheme has no reason to reduce. Replacing the relative
distortion by its normal part repairs the theory. The normal distortion
contracts in two steps with the exact constant $2697/4096$, the filtered
curvature differences contract with the constant $91/128$, an admissible
class of polygons propagates through all levels with edges halving to
third order and the safeguard never activating, the discrepancy from the
normal companion is a pure change of parameter that vanishes on
arc-uniform data, and the main theorem delivers convergence to a regular
interpolatory limit that is $G^2$ and, in arc length, $C^2$. For this
scheme, the result realises the enhancement conjectured at the end
of~\cite{ewald2015holder}, whereby regularity is governed by the normal
companion alone.

Several directions remain open, and we state them in the order in which
they seem approachable. The first is the H\"older endpoint. The
measurements reach the DD6 signature $2.82$ after reparametrisation, and
promoting the two-step sup-norm constants of Lemma~\ref{lem:linear} to
joint spectral quantities at the curvature level should yield arc-length
$C^{2,\alpha}$ regularity for every $\alpha$ below $0.830$, but the
bookkeeping of the nonlinear remainders at that precision is genuinely
harder than what Theorem~\ref{thm:invariance} required. The second is
the general form of the conjecture. Our proof uses two properties of the
scheme beyond the framework axioms, namely the exact restriction of
Theorem~\ref{thm:restriction}, which pins the tangential dynamics to
midpoint subdivision, and the fourth-order proximity of
Proposition~\ref{prop:proximity}, and a general theorem would have to
replace both by conditions stated on the companions alone. We expect the
normal-distortion calculus of Section~\ref{sec:normal} to survive that
generalisation unchanged, and the reparametrisation argument to be the
delicate part. The third direction is intrinsic. The spherical
experiments of Section~\ref{sec:spherenum} show the same stabilisation
mechanism at work, and a spherical linearisation theorem, presumably
with curvature-dependent corrections to the filter coefficients, would
put those experiments on the footing the planar theory now enjoys.
Transferring Theorems~\ref{thm:invariance} and~\ref{thm:main} to the
sphere and the hyperbolic plane appears within reach of the same
methods, with the space-form corrections entering at the order of the
remainders, in the spirit of the manifold proximity
analyses~\cite{wallner2005convergence,grohs2010general,wallner2011riemannian,wallner2014interpolatory}.
The hyperbolic construction itself requires the generalized-circle
curvature estimator for the circle, horocycle, and hypercircle family
together with an intrinsic implementation, and the Jacobi factor
argument of Theorem~\ref{thm:solvable} extends verbatim with $\sinh$ in
place of $\sin$, which suggests that a hyperbolic version exists.
Finally, a systematic robustness study on strongly irregular data,
beyond the single stress example of Section~\ref{sec:stress} and beyond
the admissible class, completes the list.

We close with two remarks. The first concerns cost. Each insertion is a
two-by-two Newton solve, with a median of one to three iterations in all
our experiments, wrapped around a segment integration, so the scheme is
slower than an affine rule by a constant factor and entirely practical.
The second concerns the shape of the theory. Geometric schemes are built
to serve the shape rather than the parametrisation, and the analysis
here suggests that their theory becomes simplest when the measuring
instruments respect the same distinction. The relative distortion is a
fine instrument that happens to measure both, and splitting it into its
tangential and normal parts is the novel idea we believe we present in
this paper.

\subsection*{Acknowledgements}
The author acknowledges the computational resources provided by the Centre
for Visual Computing and Intelligent Systems at the University of
Bradford.

\end{document}